\documentclass[12pt, draftclsnofoot, onecolumn]{IEEEtran}
\usepackage[utf8]{inputenc}
\usepackage[english]{babel}
\usepackage{cite}
\IEEEoverridecommandlockouts
\usepackage[pdftex]{graphicx}
\usepackage[font=small,labelfont=bf]{caption}
\usepackage{subcaption}
\usepackage{enumitem} 
\usepackage{cite}
\usepackage{float}
\usepackage{amsfonts}%
\usepackage{amsmath}%
\usepackage{fixmath}%
\usepackage{amssymb}%
\usepackage{graphicx}
\usepackage{accents}
\usepackage{comment}
\usepackage{url}
\usepackage{booktabs}
\usepackage[table,xcdraw]{xcolor}
\usepackage[bottom]{footmisc}
\usepackage{setspace}
\usepackage{mathtools}
\usepackage{fixltx2e}
\usepackage{stfloats}
\usepackage{tabularx}
\usepackage{amsthm}
\usepackage{mathrsfs}
\usepackage{amsbsy}

\usepackage{algorithmic}[2]
\usepackage[linesnumbered,boxed,ruled,vlined]{algorithm2e}
\newtheorem{thm}{Theorem}
\newtheorem{lem}[thm]{Lemma}

\newtheorem{defn}{Definition}

\usepackage{hyperref}

\allowdisplaybreaks

\begin{document}
	\title{Can Dynamic TDD Enabled Half-Duplex Cell-Free Massive MIMO Outperform Full-Duplex Cellular Massive MIMO?}
	\author{Anubhab Chowdhury, Ribhu Chopra, and Chandra R. Murthy
		\thanks{Anubhab Chowdhury and Chandra R. Murthy are  with  the  Department   of   ECE,   Indian   Institute   of   Science,   Bengaluru 560012,   India.   (e-mails:\{anubhabc, cmurthy\}@iisc.ac.in.)}
	\thanks{Ribhu Chopra is with the Department of EEE, Indian Institute of Technology Guwahati, Assam 781039, India (e-mail: ribhufec@iitg.ac.in.)}}

	\maketitle
	\begin{abstract}
	We consider a dynamic time division duplex~(DTDD) enabled cell-free massive multiple-input multiple-output~(CF-mMIMO) system, where each half-duplex~(HD) access point~(AP) is scheduled to operate in the uplink~(UL) or downlink~(DL) mode based on the data demands of the user equipments~(UEs), \textcolor{black}{with the goal of maximizing the sum UL-DL spectral efficiency (SE).} 
	\textcolor{black}{We develop a new, low complexity, greedy algorithm for the combinatorial AP scheduling problem, with an optimality guarantee theoretically established via showing that a lower bound of the sum UL-DL SE is sub-modular.}
	We also consider pilot sequence reuse among the UEs to limit the channel estimation overhead. In CF systems, all the APs estimate the channel from every UE, making pilot allocation problem different from the cellular case. We develop a novel algorithm that iteratively minimizes the maximum pilot contamination across the UEs.  We compare the performance of our solutions, both theoretically and via simulations, against a full duplex~(FD) multi-cell mMIMO system. Our results show that, due to the joint processing of the signals at the central processing unit, CF-mMIMO with dynamic HD AP-scheduling significantly outperforms cellular FD-mMIMO in terms of the sum SE and $90\%$ likely SE. Thus, DTDD enabled HD CF-mMIMO is a promising alternative to cellular FD-mMIMO, without the cost of hardware for self-interference suppression.
	\end{abstract}

\IEEEpeerreviewmaketitle

    \section{Introduction}
    Cell free massive multiple-input multiple-output~(CF-mMIMO) refers to a network architecture where multiple access-points~(APs) coherently and simultaneously serve a number of user equipments~(UEs) distributed over a large geographical area~\cite{cell_small, Nayebi,EGL2,EB_Stochastic}. Recently, CF-mMIMO has emerged as a promising candidate technology for the physical layer of next generation wireless communication systems~\cite{6G}. It has
    been shown that under appropriate conditions~\cite{EB_Stochastic}, CF-mMIMO inherits many of the advantages offered by cellular massive MIMO such as channel hardening and favorable propagation.  
    However, in their current form, CF-mMIMO systems are designed to work in the time division duplexed (TDD) mode, hence serving either only uplink~(UL) or only downlink~(DL) traffic at any given point in time.
While enabling full duplex~(FD) capabilities at the APs can simultaneously cater to the UL and DL data demands, the performance of such systems is limited by the residual self-interference~(SI) power at each AP~\cite{Sabharwal}.
\vspace*{-0.1cm}
\subsection{Motivation}  

In the context of cellular mMIMO, dynamic TDD (DTDD) has recently been explored to cater to heterogeneous UL-DL data demands from the UEs. This technique entails adaptive and independent splitting of the transmission frame into UL and DL slots by the different base stations (BSs) according to the UL-DL traffic demands from the UEs in each cell~\cite{DTDD_TWC}.
While this improves the overall spectral and time resource utilization  across cells, it does not fully cater to heterogeneous data demands within the cells. That is, a UE with UL data  demand will still have  to wait for a slot where its serving BS is operating in the UL mode in order to complete its transmission, and similarly for a UE with DL data demand. On the other hand, in a CF-mMIMO system, since the UEs are not associated with a particular AP, if the APs can dynamically select the slots where they operate in UL and DL modes, any UE with a specific data demand can find some nearby APs operating in the corresponding mode in the same slot. Further, the joint processing of the signals at the CPU can mitigate the cross-link interferences that arise in a CF system. Due to this, a CF-mMIMO system with DTDD can potentially match or even exceed the performance of an FD-capable cellular system, while using half-duplex (HD)  hardware at the APs. Therefore, the use of DTDD in conjunction with CF-mMIMO is the focus of this work. 
\vspace*{-0.1cm}
    \subsection{Related Work}
    DTDD is a well accepted technique; it  has been included in cellular communication standards such as $3$GPP LTE Release $12$~\cite{3GPP_config} and 5G NR~\cite{5GNR,5GNR_crosslink} to accommodate heterogeneous traffic loads.
Traffic-dependent UL-DL slot adaptation schemes have been shown to reduce the overall system latency~\cite{next_gen_TDD} and improve the spectral efficiency (SE)~\cite{DTDD_testbed,textbed_Sir} compared to TDD-based conventional cellular and CF mMIMO systems. However, the performance of DTDD is limited by two types of cross link interferences~(CLIs), namely, the interference from the DL BSs to the UL BSs and the interference from the UL UEs to the DL UEs. The CLI can  be mitigated via intra cell-cooperation, power control and beamforming design, UE scheduling, etc. An excellent survey on the methods for CLI mitigation in cellular mMIMO can be found in~\cite{survey}. 
    
    On the other hand, FD technology can also serve UL and DL UEs simultaneously and has the potential to double the system capacity. Note that, in a cellular FD mMIMO system, similar CLIs exist as in DTDD based systems. However, in addition, each BS suffers from its own residual SI. In fact, the transmit RF-chain noise, oscillator phase noise, and related device imperfections get amplified while propagating through the SI channel and limit the FD system performance~\cite{Sabharwal}. Also, the benefits of an FD cellular system considerably degrades under asymmetric traffic load~\cite{FD_hybrid}. In contrast, DTDD obviates the need for expensive and potentially power-hungry hardware as well as digital signal processing costs associated with SI mitigation. Numerical experiments have shown that the throughput of the cellular FD-system  degrades relative to cellular DTDD as the asymmetry between the UL and the DL traffic increases~\cite{FD_DTDD_1}. 
    
    The current cellular deployments of DTDD require inter-cell cooperation, i.e., the neighboring cells are required to exchange information~(such as estimated channel statistics or the per cell traffic load) for optimal UL-DL slot scheduling or interference mitigation. Although such techniques are attractive in theory, the sub-problems  of BS/UE scheduling, power control, cell clustering~\cite{5GNR_cluster}, and joint beamformer design~\cite{Beam_DTDD} are prohibitively complex to implement in a practical system. Moreover, the performance loss of the cell-edge UEs due to out-of-cell as well as CLI is a serious issue with any cellular architecture. 
    
In contrast to cellular mMIMO, in a CF system, all the UEs in a given geographical area are served by all the available APs by jointly processing the signals to/from the  UEs at a central processing unit (CPU). At the cost of a larger front-haul bandwidth, the CPU can utilize the knowledge of locally estimated channels from each AP to suppress the CLIs without inter-AP cooperation or extra signaling overhead~\cite{net_work}. Due to this, the quality of service (QoS) delivered is nearly uniform across all the UEs~\cite{EGL2}. The advantages offered by DTDD along with the inherent benefits of the CF-architecture can be exploited to further enhance the system throughput under asymmetric traffic load. We not only dispense with the SI cancellation hardware at each AP; the computational burden and signaling overhead involved in CLI mitigation of a cellular mMIMO system is also considerably reduced at the CPU. 
    
    In the context of DTDD enabled CF-mMIMO systems, the authors in~\cite{net_work} presented a UE scheduling algorithm to alleviate the CLI from UL UEs to DL UEs. Recently, in~\cite{Lee_DTDD}, the authors proposed a so-called beamforming training based scheme, where the estimates of the effective DL channels are exploited to reduce inter-AP CLI.  All the previous works assume a fixed UL and DL configuration across the APs, and focus primarily on CLI mitigation methods. However, unless the transmission and reception mode of each AP is dynamically adapted based on  the traffic demands of the UEs, the benefits of DTDD cannot be fully exploited. Therefore, to enable DTDD, we need to split the time resources optimally at each of the APs. However, scheduling the APs via exhaustive search over all possible AP configurations is prohibitively complex.  Motivated by this, we formulate the problem of optimally scheduling APs in the UL or DL modes to maximize the sum UL-DL SE in a DTDD enabled CF-mMIMO system and propose a scalable solution by exploiting a sub-modularity property of the sum UL-DL SE.
    
   We note that, in CF-mMIMO systems, a natural UE grouping by the serving BSs does not exist, unlike a cellular mMIMO system. Also, in  cellular mMIMO, only the serving BS aims to estimate the channel from a given UE, while in a CF-mMIMO system, all the APs in the vicinity of a given UE need to obtain good channel estimates. Therefore, in order to mitigate pilot contamination, one needs to revisit the problem of pilot allocation across UEs in CF-mMIMO systems. For example, physically proximal UEs should not use the same pilot sequences.  
   Therefore, in this paper, we also address the problem of pilot allocation in CF-mMIMO systems along with the problem of AP-scheduling.

\subsection{Contributions:}

In this paper, we investigate how to facilitate DTDD in a CF-mMIMO system with HD-APs. DTDD allows us to partition the time-slots at each AP into UL and DL slots according to the UL and the DL traffic demands at the UEs. The scheduling of APs based on the data demands and analyzing the resulting network throughput performance is the main goal of this work. 
Our main contributions are as follows:
    \begin{enumerate}
    	\item We formulate the AP-scheduling problem as one of maximizing the sum UL-DL SE given the traffic demands from the UEs and considering  matched filter precoding~(MFP) in the DL and maximal ratio combining~(MRC) in the UL based on the locally estimated channels.  This problem turns out to be NP-hard, and hence the computational complexity of a brute force  search based solution grows exponentially with the number of APs. We first argue that the achievable sum UL-DL SE is a monotonic nondecreasing function of the set of scheduled APs. Then, we observe that the dependence of the sum UL-DL SE on the scheduled AP-set is non-linear in nature and therefore proving sub-modularity becomes mathematically intractable. To circumvent that, we derive the following results:
    	\begin{enumerate}
    		\item  We lower bound the sum UL-DL SE and prove that problem of maximizing the lower bound is equivalent to the problem of maximizing product of the SINRs.
    		\item We prove that product of the SINRs of all UEs is a sub-modular set function of the APs scheduled in the system. (See Theorem~\ref{thm:submod}.)
    	\end{enumerate} 
     \item This allows us to develop a greedy algorithm for dynamic AP-scheduling, where,
    	at each step, the transmission mode of the AP that maximizes the incremental SE is
    	added to the already scheduled AP-subset. The lower bound on the sum UL-DL SE achieved by the solution obtained via the greedy algorithm is guaranteed to be within
    	a $(1-\frac{1}{e})$-fraction of its global optimal value.~(See Algorithm~\ref{algo:GA_submod}.)
    	We note that the computational complexity of the greedy algorithm is linear in the number of APs. 
    	\item \textcolor{black}{We also analyze the UL and DL SE considering a minimum mean square error~(MMSE) based combiner in the UL and regularized zero forcing~(RZF) precoder in the DL, and demonstrate the performance improvement obtained compared to MRC and MFP.}
    	\item We then shift focus to the problem of allocating pilot sequences across the UEs. In a CF system, optimal pilot allocation is also a combinatorial optimization problem, and unless the pilots are judiciously assigned across UEs, pilot contamination can lead to poor channel estimates at multiple APs. Therefore, we develop an iterative pilot allocation algorithm based on locally estimated channel statistics~(see Algorithm~\ref{GA_pilot}) at the APs. This algorithm  does not require extra signaling overhead in the form of inter-AP coordination.~(See Sec.~\ref{sec: pilot_allocation}.)

    \end{enumerate}

\textcolor{black}{Our experimental results show that the greedy algorithm procures a sum UL-DL SE that matches with exhaustive search based AP-scheduling, and that the algorithm is robust to both inter-UE and inter-AP CLI.
    	Furthermore, DTDD CF-mMIMO substantially enhances the system  performance compared to  static TDD based CF as well as cellular systems. 
	Interestingly, the DTDD based CF-system outperforms an FD cellular mMIMO system under both MRC \& MFP as well as MMSE \& RZF combiner and precoder employed at the APs/BSs. For example, a CF-DTDD system with $(M=16,N=64)$ even outperforms the cellular FD-system having twice the antenna density, i.e., $(L=16, N_t=N_r=64)$. If we increase the number of APs with half the antenna density compared to the FD (see the curve corresponding to $(M=64, N=16)$ in Fig.~\ref{fig:CDF_sum_SE}), the sum UL-DL SE offered by HD CF-DTDD improves, significantly outperforming the cellular FD system. }

We conclude that, due to the benefits offered by joint signal processing at the CPU, HD CF-mMIMO with dynamic AP-scheduling  offers improved sum SE as well as $90\%$-likely SE compared to static TDD based CF systems and even the cellular FD mMIMO system. Therefore, DTDD enabled CF-mMIMO system with appropriately scheduled APs is a promising solution to meet the heterogeneous traffic loads in next generation wireless systems.

\emph{Notations:} 
    Matrices, vectors, and sets are denoted by bold upper-case, bold lower-case, and calligraphic letters, respectively. $(\cdot)^T$, $(\cdot)^H$, $(\cdot)^*$, and $\text{tr}(\cdot)$ represent transposition, hermitian, complex conjugation, and trace operations, respectively.  $| \cdot |$, $\setminus$, $'$, and $\cup$ denote the cardinality, set-subtraction, complement, and union of sets, respectively. $\mathbb{E}[\cdot]$ and ${\tt var}\{\cdot\}$ denote the mean and variance of a random variable, respectively. $\mathbf{x}\sim\mathcal{CN}(\mathbf{0},\boldsymbol{\Sigma}_{N})$ indicates that $\mathbf{x}$ is a zero mean circularly symmetric complex Gaussian (CSCG) random vector with  covariance matrix $\boldsymbol{\Sigma}_N\in\mathbb{C}^{N\times N}$. \textcolor{black}{
    Other frequently used symbols  are cataloged in Table~\ref{Table: Symbols}.}
    
\begin{table}[]
	\centering
	\caption{\textcolor{black}{Symbols}}
	\label{Table: Symbols}
	\resizebox{\textwidth}{!}{%
		\begin{tabular}{@{}|
				>{\columncolor[HTML]{EFEFEF}}l |l|
				>{\columncolor[HTML]{EFEFEF}}l |l|@{}}
			\toprule
			\cellcolor[HTML]{CEF5CE}{\color[HTML]{000000} \textbf{Symbol}} &
			\cellcolor[HTML]{CEF5CE}\textbf{Definition} &
			\cellcolor[HTML]{CEF5CE}\textbf{Symbol} &
			\cellcolor[HTML]{CEF5CE}{\color[HTML]{000000} \textbf{Definition}} \\ \midrule
			$\mathbf{A}\in\mathbb{C}^{m\times n}$ &
			A matrix with $m$ rows and $n$ columns &
			$\mathbf{a}\in\{0,1\}^{n}$ &
			An $n$ length vector with each element being either $0$ or $1$ \\ \midrule
			$\mathbf{I}_{N}$ &
			Identity matrix of dimension $N\times N$ &
			$|\mathcal{A}|=n$ &
			$\mathcal{A}$ is an index set with cardinality $n$ \\ \midrule
			\begin{tabular}[c]{@{}l@{}}$\mathcal{A}(m)$\end{tabular} &
			\begin{tabular}[c]{@{}l@{}} $m$th element of the index set $\mathcal{A}$\end{tabular} &
			$\{\boldsymbol{\phi}_1,\boldsymbol{\phi}_2,\ldots,\boldsymbol{\phi}_{\tau_p}\}$&
			\begin{tabular}[c]{@{}l@{}} The set of $\tau_{p}$ orthonormal pilot sequences\end{tabular} 
			\\ \midrule
			\begin{tabular}[c]{@{}l@{}}$\mathcal{I}_{p}$\end{tabular} &
			\begin{tabular}[c]{@{}l@{}}The set of UEs that use the $p$th pilot sequence\end{tabular}&
			$\beta_{mk}$ &
			Pathloss coefficient between the $m$th AP and $k$th UE
			\\ \midrule
			$\mathbf{h}_{mk}$ &
			Fast fading channel between the $m$th AP and $k$th UE &
			$\alpha_{mk}^2$&
		The variance of the $k$th UE's estimated channel at the $m$th AP\\ \midrule
			$\mathcal{E}_{p,k}$ &
			Pilot power of the $k$th UE 
			 &
			$\mathcal{E}_{u,k}$ and $\mathcal{E}_{d,j}$&				\begin{tabular}[c]{@{}l@{}}
			UL data power of the $k$th UE and DL data power of the $j$th AP\end{tabular}
		 \\ \midrule
			$\mathcal{U}$ &
			Set of all UEs, with  $|\mathcal{U}|=K$&
			$\mathcal{A}$ &
			Set of all AP indices \\ \midrule
			$\mathcal{U}_{u}$ and $\mathcal{U}_{d}$ &
			Set of UL and DL UEs, respectively &
			$\kappa_{jk}$ &
			DL power control coefficient for the $k$th UE at the $j$th AP  \\ \midrule
			$\mathcal{A}_{u}$ and $\mathcal{A}_{d}$&\begin{tabular}[c]{@{}l@{}}
			Set of UL and DL APs, respectively\end{tabular} &
			$M$ and $N$&
			Number of APs and antennas per AP, respectively \\ \midrule
			$\epsilon_{nk}$ &
			\begin{tabular}[c]{@{}l@{}}The variance of the channel between the $n$th UL UE\\  and the $k$th DL UE\end{tabular} &
			$\zeta_{mj}$ &
			\begin{tabular}[c]{@{}l@{}}The variance of the residual interference channel between \\ the $j$th DL AP and the $m$th UL AP\end{tabular} \\ \bottomrule
		\end{tabular}%
	}
\vspace*{-.7cm}
\end{table}

	\section{System Model and Problem Statement}

	We consider a CF-mMIMO system with $M$ HD-APs jointly and coherently serving $K$ single-antenna UEs. Each AP is equipped with $N$ antennas and is connected to the CPU via an infinite capacity front-haul link. Time is divided into slots, and in any given slot, each AP can operate 
	either in the UL mode or in the DL mode. 
	We assume that the UL/DL traffic demands of the UEs are known at the CPU; its task is to decide the mode of operation of each AP based on the traffic demands in its vicinity.
	
	The channel from $k$th UE to the $m$th AP is modeled as $\mathbf{f}_{mk} =\sqrt{\beta_{mk}}\mathbf{h}_{mk} \in \mathbb{C}^N$, where $\beta_{mk} > 0$ denotes the large scale fading and path loss coefficient, and \textcolor{black}{are known} to the APs and the CPU. \textcolor{black}{Note that $\beta_{mk}$ remains unchanged over several channel coherence intervals\cite{cell_small,Nayebi,EGL2}.} The fast fading components, $\mathbf{h}_{mk}  \sim \mathcal{CN}(\mathbf{0},\mathbf{I}_N) \in \mathbb{C}^N$, are the independent and identically distributed~(i.i.d.) and \textcolor{black}{are estimated at  the APs (and the CPU) using pilot signals. Under a quasi-static fading model, $\mathbf{h}_{mk}$ remains constant over one coherence interval, and takes independent values from the same distribution in subsequent coherence intervals~\cite{cell_small,Nayebi,EGL2,Clustered_CF,Pilot_sum_rate,random_vs_structured,making_cf_emil}. } Furthermore, the foregoing analysis can be extended to the case of spatially correlated channels with some effort, but the equations become more cumbersome and do not offer significant additional insights.

	Due to simultaneous UL and DL data transmissions, the APs transmitting in the DL cause interference to the APs receiving the UL data, which is the source of inter-AP interference. 
    However, since the channel state information (CSI) of the inter AP channels available at the CPU may be erroneous, residual inter AP interference may exist even after interference cancellation. In this paper, we account for the effect of residual interference due to imperfect AP-to-AP interference cancelation in our analysis. \textcolor{black}{In the literature, the residual SI is modeled as Gaussian distributed independent additive noise~\cite{Asymptotic,net_work,FD_CF_ICC,my_FD_SPAWC,AliC}; we use the same approach. We model the residual interference channel between $j$th DL AP and the $m$th UL AP by $\mathbf{{G}}_{mj}\in \mathbb{C}^{N\times N}$, with its elements being i.i.d. $\mathcal{CN}(0,\zeta_{mj})$, where $\zeta_{mj}$ depends on the inter AP path loss and channel estimation error variance.}
	Similarly, we let $\mathtt{g}_{nk}$ denote the channel between $n$th UL UE and the $k$th DL UE, and we model $\mathtt{g}_{nk} \sim \mathcal{CN}(0,\epsilon_{nk})$ and  independent across all UEs~\cite{my_FD_SPAWC,Asymptotic}. 
	\vspace{-0.1cm}
	
	\subsection{Problem Statement}\label{sub_sec_problem}
		In this work, we investigate how to enable DTDD in a CF mMIMO system with HD APs. 
		Let $\mathcal{A}$ be the set of AP indices, with $M=|\mathcal{A}|$. Let the indices of the APs scheduled in the UL and DL modes be contained in the index sets $\mathcal{A}_{u}$ and $\mathcal{A}_{d}$, respectively. 
		We aim to  maximize the achievable sum UL-DL SE $R_{\mathrm{sum}}(\mathcal{A}_{u}\cup\mathcal{A}_{d})$ over all possible choices of $\mathcal{A}_{u}$ and $\mathcal{A}_{d}$ by solving:		\vspace{-0.1cm}
		\begin{align}\label{eq:problem_statement}
			&\textstyle{\max_{\mathcal{A}_{u},\mathcal{A}_{d}}\quad R_{\mathrm{sum}}(\mathcal{A}_{u}\cup\mathcal{A}_{d})}\notag\\&\textstyle{\mathrm{s.t.}\quad \mathcal{A}_{u},\mathcal{A}_{d}\subset\mathcal{A}, \quad \mathcal{A}_{u}\cap\mathcal{A}_{d}=\emptyset, \quad  \mathcal{A}_{u}\cup\mathcal{A}_{d}=\mathcal{A}.}
		\end{align}
In the above, the condition $\mathcal{A}_{u}\cap\mathcal{A}_{d}=\emptyset$ arises because of the half-duplex constraint at the APs. Evidently, an exhaustive search can be performed across all $2^{M}$ possible configurations, but this becomes computationally expensive as $M$ gets large. We develop a low complexity AP scheduling algorithm in Sec.~\ref{sec:Submodular}.
			
Another aspect that we address in the current work relates to the problem of channel estimation. As mentioned earlier, in CF-mMIMO, all APs are potentially interested in estimating the channel from all UEs, as the signals are jointly processed at the CPU. We consider the reuse of a set of orthonormal pilot sequences across the UEs in order to limit the pilot overhead. However, in this case, it is necessary to assign pilots to UEs such that, at any given AP, the UEs in its vicinity are assigned orthogonal pilots, as far as possible. We address this problem in Sec.~\ref{sec: pilot_allocation}. 
			
We next discuss the CF-mMIMO signal model within each slot and the channel estimation at the APs, which is our point of departure in this work.

	\section{Signaling Model and Channel Estimation}\label{sec: channel_estimation}
	We assume that each slot consists of $\tau$ symbols or channel uses, of which the first $\tau_{p}$ are used for the UL channel estimation. In these $\tau_p$ symbols, all the UEs transmit pilot sequences to the APs. 
	Due to the large number of UEs being served, the training overhead associated with allotting orthogonal pilot sequences to all $K$ UEs could be inordinately high, as it requires $\tau_{p}\ge K$. Hence, we use a set of $\tau_p$ orthonormal pilot sequences denoted by $\{\boldsymbol{\phi}_1,\boldsymbol{\phi}_2,\ldots,\boldsymbol{\phi}_{\tau_p}\}$, where $\tau_{p}$ could be less than $K$, and $\boldsymbol{\phi}_p\in\mathbb{C}^{\tau_p},  p=1,\ldots,\tau_{p}$. 
	These $\tau_{p}$ sequences are allotted to the $K$ UEs, and we denote the indices of the UEs employing the $p$th pilot sequence  by the set $\mathcal{I}_p$. Note that the cardinality of $\mathcal{I}_p$ indicates the repetition factor of the $p$th pilot sequence, such that $\sum_{p=1}^{\tau_{p}}|\mathcal{I}_p|=K$.  
	Let $\mathcal{E}_{p,k}$ be the power of the pilot signal by the $k$th UE. Then, the received pilot signal matrix at the $m$th AP becomes
	$\textstyle{\mathbf{Y}_{p,m}=\sum\nolimits_{p=1}^{\tau_{p}}\sum\nolimits_{k\in \mathcal{I}_p}\sqrt{\tau_{p}\mathcal{E}_{p,k}}\mathbf{f}_{mk}\boldsymbol{\phi}_{p}^{T}+\mathbf{W}_{p,m}\in\mathbb{C}^{N\times \tau_{p}}},$
where $\mathbf{W}_{p,m}$ is the receiver noise matrix whose columns are distributed as $\mathcal{CN}(\mathbf{0},N_{0}\mathbf{I}_N)$. Post-multiplying $\mathbf{Y}_{p,m}$
by $\boldsymbol{\phi}_{p}^{*}$, the processed signal at the $m$th AP becomes
	\begin{equation}\label{eq:pilot_receieve}
	\textstyle{\mathbf{\dot{y}}_{p,m}=\sqrt{\mathcal{E}_{p,k}\tau_{p}}\mathbf{f}_{mk}+\sum\nolimits_{n\in \mathcal{I}_{p}\setminus k}\sqrt{\mathcal{E}_{p,n}\tau_{p}}\mathbf{f}_{mn}+\dot{\mathbf{w}}_{p,m}},
	\end{equation}
	with $\dot{\mathbf{w}}_{p,m}=\dot{\mathbf{W}}_{p,m}\boldsymbol{\phi}_{p}^{*}~\sim\mathcal{CN}(\mathbf{0},N_{0}\mathbf{I}_N)$. We can now estimate the channel $\mathbf{f}_{mk}$ using $\mathbf{\dot{y}}_{p,m}$. 
	We have the following result from \cite{my_MMSE_SPAWC}. 
The linear minimum mean squared error~(LMMSE) estimate of $\mathbf{f}_{mk}$, denoted as $\hat{\mathbf{f}}_{mk}$, can be evaluated as $\hat{\mathbf{f}}_{mk}={\mathbb{E}[\mathbf{f}_{mk}\mathbf{\dot{y}}_{p,m}^{H}]}({\mathbb{E}[\mathbf{\dot{y}}_{p,m}\mathbf{\dot{y}}_{p,m}^{H}]})^{-1}\mathbf{\dot{y}}_{p,m}$, which becomes $\textstyle{\hat{\mathbf{f}}_{mk}=\frac{\sqrt{\tau_{p}\mathcal{E}_{p,k}}\beta_{mk}}{\tau_{p}\mathcal{E}_{p,k}\beta_{mk}+\tau_{p}\sum\nolimits_{n\in \mathcal{I}_{p}\backslash k}{\mathcal{E}_{p,n}\beta_{mn}}+N_{0}}\mathbf{\dot{y}}_{p,m}}$,
and $\hat{\mathbf{f}}_{mk}{\sim}\mathcal{CN}(\mathbf{0},\alpha_{mk}^2\mathbf{I}_{N})$, with $\alpha_{mk}^2=c_{mk}\tau_{p}\mathcal{E}_{p,k}\beta_{mk}^2$, and $c_{mk}\triangleq({\tau_{p}\mathcal{E}_{p,k}\beta_{mk}+\tau_{p}\sum\nolimits_{n\in \mathcal{I}_{p}\backslash k}{\mathcal{E}_{p,n}\beta_{mn}}+N_{0}})^{-1}$.
The estimation error, denoted by $\tilde{\mathbf{f}}_{mk} \triangleq {\mathbf{f}}_{mk} - \hat{\mathbf{f}}_{mk}$, is distributed as $\mathcal{CN}(\mathbf{0},\bar{\alpha}_{mk}^{2}\mathbf{I}_{N})$, with $\bar{\alpha}_{mk}\triangleq\sqrt{\beta_{mk}-\alpha_{mk}^2}$. 

    \subsection{Data Transmission}
Let $\mathcal{U}_{u}, \mathcal{U}_{d}, \mathcal{A}_{u}$, and $\mathcal{A}_{d}$ denote the sets containing the indices of UEs demanding UL data, UEs demanding DL data, UL APs, and DL APs, respectively. Now, let the $k$th UL UE send the symbol $s_{u,k}$ with power  $\mathcal{E}_{u,k}$. The data symbol of each UE is assumed to be zero mean, unit variance, and independent of the data symbols sent by the other UEs. Then, the UL signal received  at the $m$th AP~($m\in{\mathcal{A}_{u}}$) can be expressed as
    \begin{align}\label{eq:UL_receieved}
    \textstyle{\mathbf{y}_{u,m}=\sqrt{\mathcal{E}_{u,k}}\mathbf{f}_{mk}s_{u,k}+\sum\nolimits_{\substack{n\in\mathcal{U}_{u}\backslash k}}{\textstyle\sqrt{\mathcal{E}_{u,n}}}\mathbf{f}_{mn}s_{u,n}+\sum\nolimits_{j\in \mathcal{A}_{d}} \mathbf{G}_{mj}\mathbf{x}_{d,j}+\mathbf{w}_{u,m}\in \mathbb{C}^{N}},
    \end{align}
    where $\mathbf{w}_{u,m} \sim \mathcal{CN}(\mathbf{0},N_{0}\mathbf{I}_N)$ is the additive noise, and $\mathbf{x}_{d,j}=\sqrt{\mathcal{E}_{d,j}}\mathbf{P}_{j}\text{diag}(\boldsymbol{\kappa}_{j})\mathbf{s}_{d}$ is the transmitted DL data vector, 
    with  $\mathcal{E}_{d,j}$ being the total power, $\mathbf{P}_{j}\in\mathbb{C}^{N\times K_{d}}$ being the precoding matrix, and $\boldsymbol{\kappa}_{j}$ being the vector of power control coefficients, all at the $j$th DL AP. Note that  $\kappa_{jn}$, i.e., the $n$th element of $\boldsymbol{\kappa}_{j}$, indicates the fraction of power dedicated by the $j$th AP to the the $n$th DL UE ($n\in\mathcal{U}_{d}$). Typically, $\kappa_{jn}$ is designed such that $\mathbb{E}\left[\|\mathbf{x}_{d,j}\|^{2}\right]\leq\mathcal{E}_{d,j}\Rightarrow\sum_{n\in\mathcal{U}_{d}}\kappa_{jn}\textcolor{black}{\mathbb{E}\|\mathbf{p}_{jn}\|^{2}}\leq 1$~\cite{cell_small}.
     
    Let $\mathbf{v}_{mk}$ denote the locally available combining vector at the $m$th AP for the $k$th UE's UL data stream. The $k$th component of the accumulated signal  received at the CPU, ${r}_{u,k} \triangleq \sum\nolimits_{m\in\mathcal{A}_{u}}\mathbf{v}_{mk}^{H}\mathbf{y}_{u,m}$, can be expanded as
    \begin{align}\label{eq: CPU_uplink_receieve}
    &\textstyle{{r}_{u,k}=\sum\nolimits_{m\in\mathcal{A}_{u}}\sqrt{\mathcal{E}_{u,k}}\mathbf{v}_{mk}^{H}\mathbf{f}_{mk}s_{u,k}+\sum\nolimits_{m\in\mathcal{A}_{u}}\sum\nolimits_{n\in\mathcal{U}_{u}\backslash k}\sqrt{\mathcal{E}_{u,n}}\mathbf{v}_{mk}^{H}\mathbf{f}_{mn}s_{u,n}}\notag\\&\textstyle{+\sum\nolimits_{j\in\mathcal{A}_{d}}\sqrt{\mathcal{E}_{d,j}}\sum\nolimits_{n\in\mathcal{U}_{d}}\kappa_{jn}\sum\nolimits_{m\in\mathcal{A}_{u}}\mathbf{v}_{mk}^{H} \mathbf{G}_{mj}\mathbf{p}_{jn}s_{d,n}+\sum\nolimits_{m\in\mathcal{A}_{u}}\mathbf{v}_{mk}^{H}\mathbf{w}_{u,mk}},
    \end{align}
    where $\mathbf{p}_{jn}$ is the $n$th column of $\mathbf{P}_{j}$.
    Similarly, assuming perfect channel reciprocity, the signal received by the $n$th~($n\in\mathcal{U}_{d}$) DL UE can be expressed as
    \begin{align}\label{eq: CPU_DL_receieve}
     \hspace*{-.48cm}\textstyle{r_{d,n}=\sum\limits_{j\in\mathcal{A}_{d}}\hspace*{-.1cm}\kappa_{jn}\sqrt{{\mathcal{E}_{d,j}}}\mathbf{f}_{jn}^{T}\mathbf{p}_{jn}s_{d,n}+\sum\limits_{k\in\mathcal{U}_{u}}\sqrt{\mathcal{E}_{u,n}}\mathtt{g}_{nk}s_{u,k}+\sum\limits_{j\in\mathcal{A}_{d}}\sum\limits_{\substack{ q\in\mathcal{U}_{d}\backslash n}}\kappa_{jq}\sqrt{{\mathcal{E}_{d,j}}}\mathbf{f}_{jn}^{T}\mathbf{p}_{jq}s_{d,q}+w_{d,n}},
    \end{align}
     where $w_{d,n}\sim\mathcal{CN}(0,N_{0})$ is the additive noise at the $n$th DL UE. 
     \begin{figure} 
     	\begin{subfigure}{0.45\textwidth}
     		\centering
     		\includegraphics[width=0.7\textwidth,height=2.5cm]{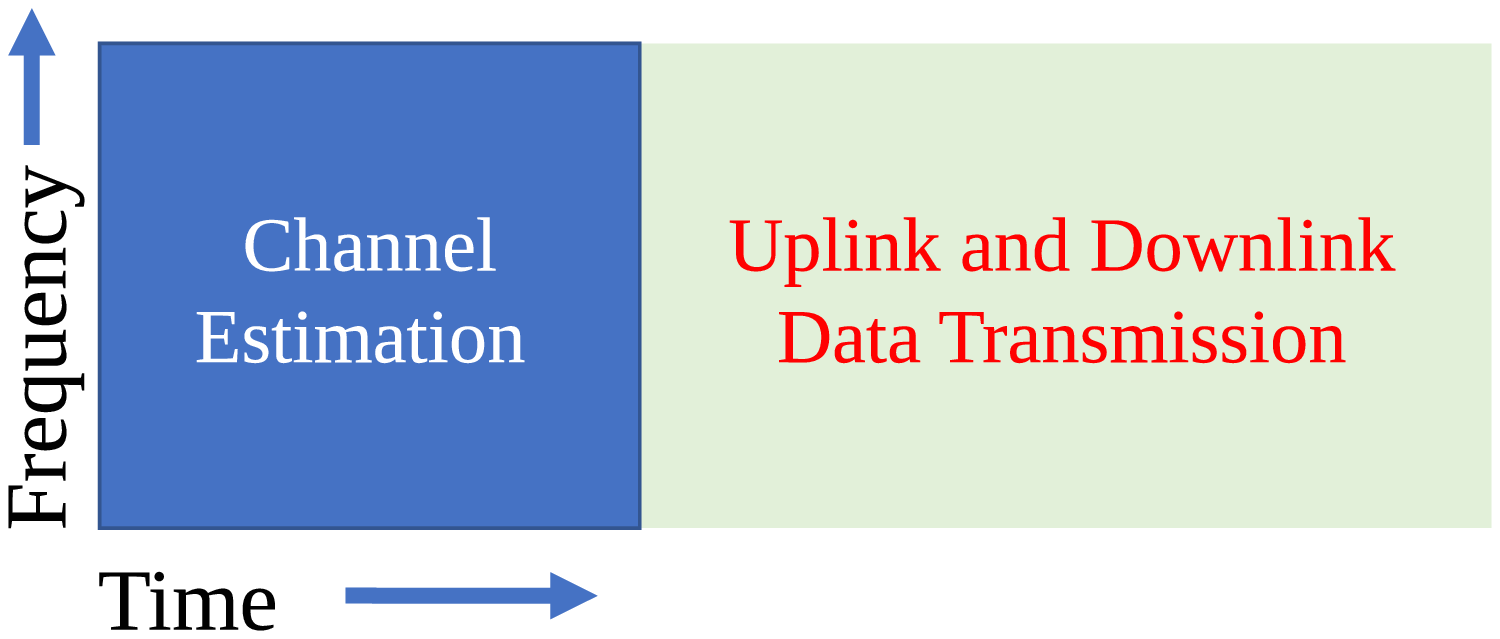}  
     		 \caption{Frame structure of DTDD based CF system}\label{fig:DTDD_frame}  	
     	\end{subfigure}\hfill
     	\begin{subfigure}{0.45\textwidth}
     		\centering
     	\includegraphics[width=0.7\textwidth,height=2.5cm]{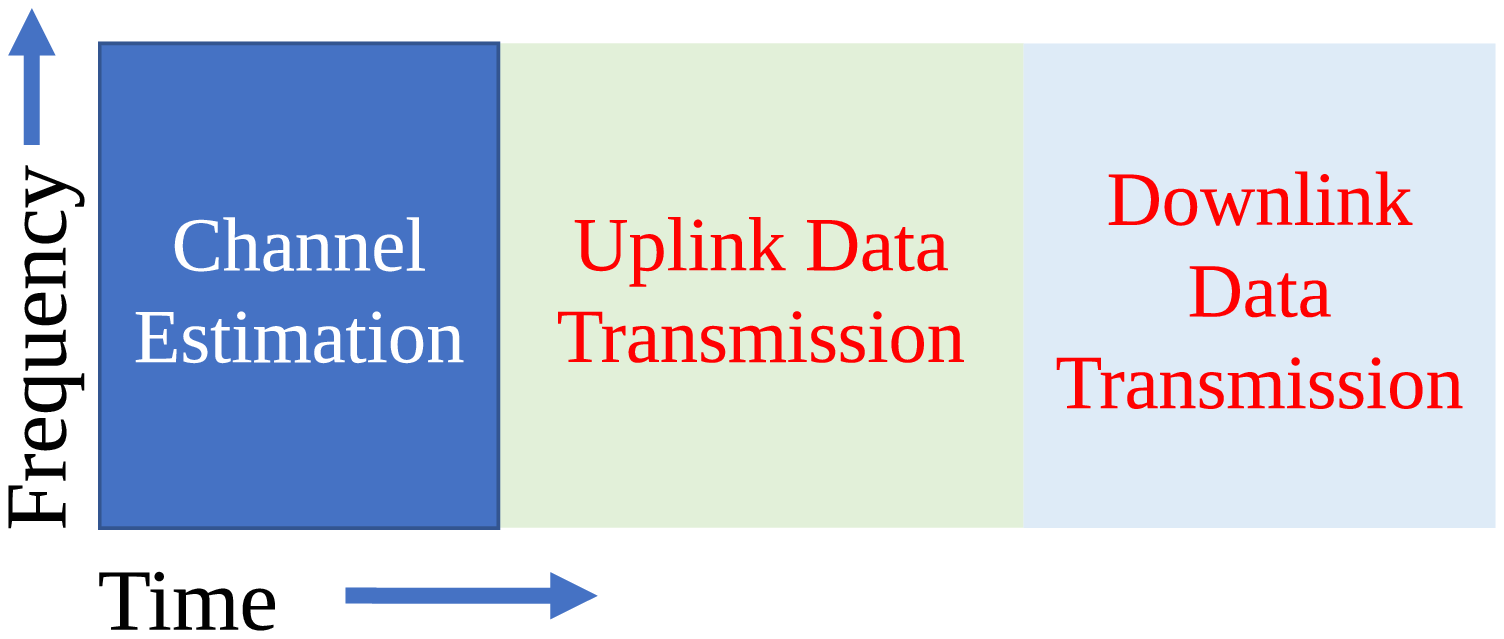}   \caption{Frame structure of TDD based CF system}\label{fig:TDD_frame}  	
     \end{subfigure}
 \caption{
 	\textcolor{black}{DTDD utilizes same time frequency resources for simultaneous UL and DL data transmission by different HD UEs/APs, unlike TDD, where time is partitioned between the UL and DL UEs.}}\label{fig:DTDD_TDD_frame}
 \vspace*{-1cm}
     \end{figure}
 
     \textcolor{black}{
     We illustrate the frame structure described above, and contrast it with the frame structure in a TDD based CF system, in Fig.~\ref{fig:DTDD_TDD_frame}. 
 } We can now derive the achievable UL and DL SEs for the DTDD enabled CF-mMIMO system. 
        \section{Spectral Efficiency Analysis: MRC \& MFP}\label{sec:SE_analysis}
   
     In this section, we derive the achievable SEs considering MFP in the DL and MRC in the UL.  Here, we first consider MFP and MRC for ease of exposition, and also because it suffices to elucidate the main point of this work, namely, the benefits obtainable by enabling DTDD in a CF-mMIMO system. In several other works, for example in~\cite{cell_small,uplink_maxmin, imrpove_cf}, MRC and MFP have been extensively used for the tractable and interpretable analysis.  For deriving the UL and DL SE, we employ the \emph{use-and-then-forget} capacity bounding technique whose effectiveness in CF-mMIMO systems has been well established~\cite{cell_small, uplink_maxmin, EGL2}.  
     
     Now, with $\mathbf{v}_{mk}=\hat{\mathbf{f}}_{mk}$ in~\eqref{eq: CPU_uplink_receieve}, the $k$th UE's of the combined  signal at the CPU becomes
  \begin{align}\label{eq: MRC_uplink}
  &\textstyle{{r}_{u,k}=\sqrt{\mathcal{E}_{u,k}}\mathbb{E}\big[\sum\nolimits_{m\in\mathcal{A}_{u}}\hat{\mathbf{f}}_{mk}^{H}{\mathbf{f}}_{mk}\big]s_{u,k}+\sqrt{\mathcal{E}_{u,k}}\big\{\sum\nolimits_{m\in\mathcal{A}_{u}}\hat{\mathbf{f}}_{mk}^{H}{\mathbf{f}}_{mk}-\mathbb{E}\big[\sum\nolimits_{m\in\mathcal{A}_{u}}\hat{\mathbf{f}}_{mk}^{H}{\mathbf{f}}_{mk}\big]\big\}s_{u,k}}\notag\\&\textstyle{+\sum\nolimits_{m\in\mathcal{A}_{u}}\hat{\mathbf{f}}_{mk}^{H}\big\{\sum\nolimits_{\substack{n\in \mathcal{I}_{p}\backslash k}}\sqrt{\mathcal{E}_{u,n}}\mathbf{f}_{mn}s_{u,n}+\sum\nolimits_{{q\in \mathcal{U}_u\backslash\mathcal{I}_{p}}}\sqrt{\mathcal{E}_{u,q}}\mathbf{f}_{mq}s_{u,q}\big\}}\notag\\&\textstyle{+\sum\nolimits_{m\in\mathcal{A}_{u}}\sum\nolimits_{j\in\mathcal{A}_{d}}\sum\nolimits_{n\in\mathcal{U}_{d}}\kappa_{jn}\sqrt{\mathcal{E}_{d,j}} \hat{\mathbf{f}}_{mk}^{H}\mathbf{G}_{mj}\hat{\mathbf{f}}_{jn}^{*}s_{d,n}+\sum\nolimits_{m\in\mathcal{A}_{u}}\hat{\mathbf{f}}_{mk}^{H}\mathbf{w}_{u,mk}.}
  \end{align}

The first and second terms of~\eqref{eq: MRC_uplink} respectively represent the expected effective array gain and UL beamforming uncertainty, and are uncorrelated with each other. Similarly, the first term is uncorrelated with all the other terms of~\eqref{eq: MRC_uplink}. Invoking the worst case noise theorem~\cite{massivemimobook}, 
the effective SINR of the $k$th UE's data stream, denoted by $\eta_{u,k}$, can be written as 
    	\begin{align}\label{eq:R_up_ergodic}
    	&\textstyle{\eta_{u,k}={{\mathcal{E}_{u,k}}\big[\big|\mathbb{E}\big[\sum_{m\in\mathcal{A}_{u}}\hat{\mathbf{f}}_{mk}^{H}{\mathbf{f}}_{mk}\big]\big|^{2}\Big]}\times\big({\mathcal{E}_{u,k}}{\tt var}\big\{\sum\limits_{m\in\mathcal{A}_{u}}\hat{\mathbf{f}}_{mk}^{H}{\mathbf{f}}_{mk}\big\}+\hspace{-.2cm}\sum\limits_{\substack{n\in \mathcal{I}_{p}\backslash k}}{\mathcal{E}_{u,n}}\mathbb{E}\big[\big|\sum\limits_{m\in\mathcal{A}_{u}}\hat{\mathbf{f}}_{mk}^{H}\mathbf{f}_{mn}\big|^{2}\big]+}\notag\\&\hspace{-.15cm}\textstyle{\sum\limits_{\substack{q\in\mathcal{U}_u \backslash\mathcal{I}_{p}}}\hspace{-.4cm}{\mathcal{E}_{u,n}}\mathbb{E}\big[\big|\!\!\sum\limits_{m\in\mathcal{A}_{u}}\!\!\hat{\mathbf{f}}_{mk}^{H}\mathbf{f}_{mq}\big|^{2}\Big]\hspace{-.1cm}+\hspace{-.15cm}\sum\limits_{n\in\mathcal{U}_{d}}\mathbb{E}\big[\big|\sum\limits_{m\in\mathcal{A}_{u}}\sum\limits_{j\in\mathcal{A}_{u}}\hspace{-.15cm}\sqrt{\mathcal{E}_{d,j}}\kappa_{jn} \hat{\mathbf{f}}_{mk}^{H}\mathbf{G}_{mj}\hat{\mathbf{f}}_{jn}^{*}\big|^{2}\big]\hspace{-.1cm}+\hspace{-.1cm}N_{0}\hspace{-.15cm}\sum\limits_{m\in\mathcal{A}_{u}}\hspace{-.15cm}\mathbb{E}\big\|\hat{\mathbf{f}}_{mk}^{H}\big\|^{2}\big)^{-1}}\hspace{-.1cm}.
    	\end{align}
We simplify the above expression in the following theorem. 
   \begin{thm}\label{thm: rate_UL_thm}
   	The achievable UL SE for the $k$th UE can be expressed as $\mathcal{R}_{u,k}=\log_{2}(1+\eta_{u,k})$, where $\eta_{u,k}$ is the UL SINR which is given by
   	\begin{align}\label{eq:rate_final_uplink}
   	\eta_{u,k}=\dfrac{N\mathcal{E}_{u,k}\left(\sum_{m\in\mathcal{A}_{u}}\alpha_{mk}^2\right)^{2}}{{\tt NCoh}_{u,k}+{\tt Coh}_{u,k}+{\tt IAP}_{u,k}+N_{0}\sum_{m\in\mathcal{A}_{u}}\alpha_{mk}^2},
   	\end{align}
   where $\alpha^2_{mk}$ is as defined after \eqref{eq:pilot_receieve}, $ {\tt NCoh}_{u,k}$ represents the non-coherent inter UE interference, $ {\tt Coh}_{u,k}$ represents the coherent inter UE interference due to pilot contamination, ${\tt IAP}_{u,k}$ represents the inter AP interference, and $N_{0}\sum_{m\in\mathcal{A}_{u}}\alpha_{mk}^2$ corresponds to the effect of AWGN in the UL. These are expressed as   
   	\begin{subequations}
\begin{align}
&\textstyle{{\tt NCoh}_{ u,k}=\sum\nolimits_{\substack{n\in\mathcal{U}_{u}}}\mathcal{E}_{u,n}\sum\nolimits_{m\in\mathcal{A}_{u}}\alpha_{mk}^2\beta_{mn}}\label{eq:NCoh}
\\&\textstyle{{\tt Coh}_{u,k}=N\sum\nolimits_{n\in \mathcal{I}_{p}\backslash k}\mathcal{E}_{u,n}\big(\sum\nolimits_{m\in{\mathcal{A}}_{u}}\alpha_{mk}^2\textstyle{\sqrt{\frac{\mathcal{E}_{p,n}}{\mathcal{E}_{p,k}}}\frac{\beta_{mn}}{\beta_{mk}}}\big)^{2}}\label{eq:Coh}
\\&\textstyle{{\tt IAP}_{u,k}=N\sum\nolimits_{m\in\mathcal{A}_{u}}\sum\nolimits_{j\in\mathcal{A}_{d}}\sum\nolimits_{n\in\mathcal{U}_{d}}\kappa_{jn}^2\zeta_{mj}\alpha_{mk}^2\alpha_{jn}^2\mathcal{E}_{d,j}}\label{eq:IIAP}.
\end{align}
   	\end{subequations}
   \end{thm}
\begin{proof}
	See Appendix~\ref{appen: rate_uplink_thm}.
\end{proof}

We now consider the DL case.  Letting $\mathbf{p}_{jn}=\hat{\mathbf{f}}_{jn}^{*}$, the signal received by the $n$th~($n\in\mathcal{U}_{d}$) DL UE can be expressed as
 \begin{align}\label{eq: user_DL_receieve}
&\hspace{-.3cm}\textstyle{r_{d,n}=\sum\nolimits_{j\in\mathcal{A}_{d}}\kappa_{jn}\sqrt{{\mathcal{E}_{d,j}}}\mathbb{E}\big[\mathbf{f}_{jn}^{T}\hat{\mathbf{f}}_{jn}^{*}\big]s_{d,n}+\sum\nolimits_{j\in\mathcal{A}_{d}}\kappa_{jn}\sqrt{{\mathcal{E}_{d,j}}}\left\{\mathbf{f}_{jn}^{T}\hat{\mathbf{f}}_{jn}^{*}-\mathbb{E}\big[\mathbf{f}_{jn}^{T}\hat{\mathbf{f}}_{jn}^{*}\big]\right\}s_{d,n}}\notag\\&\hspace{-.3cm}\textstyle{+\!\sum\limits_{j\in\mathcal{A}_{d}}\!\big\{\!\sum\limits_{\substack{ q\in \mathcal{I}_{p}\backslash n}}\!\!\sqrt{{\mathcal{E}_{d,j}}}\kappa_{jq}\mathbf{f}_{jn}^{T}\hat{\mathbf{f}}_{jq}^{*}s_{d,q}\!+\!\!\sum\limits_{{q\in\mathcal{U}_d\backslash \mathcal{I}_{p}}}\!\!\sqrt{{\mathcal{E}_{d,j}}}\kappa_{jq}\mathbf{f}_{jn}^{T}\hat{\mathbf{f}}_{jq}^{*}s_{d,q}\big\}\!+\!\sum\limits_{k\in\mathcal{U}_{u}}\!\sqrt{\mathcal{E}_{u,n}}\mathtt{g}_{nk}s_{u,k}\!+\!w_{d,n}.}
\end{align}
We present the DL SE in the following theorem.
\begin{thm}\label{thm:rate_final_dllink}
   	The achievable DL SE for the $n$th UE can be expressed as $\mathcal{R}_{d,n}=\log_{2}(1+\eta_{d,n})$, with DL SINR of the $n$th UE, $\eta_{d,n}$, expressed as
\begin{equation}\label{eq:rate_final_dllink}
\eta_{d,n}=\frac{N^2\left(\sum_{j\in\mathcal{A}_{d}}\kappa_{jn}\sqrt{\mathcal{E}_{d,j}}{\alpha_{jn}^2}\right)^{2}}{{\tt NCoh}_{d,n}+{\tt Coh}_{d,n}+{\tt IU}_{d,n}+N_{0}},
\end{equation}
where  $\alpha^2_{jn}$ is as defined after \eqref{eq:pilot_receieve}, ${\tt NCoh}_{d,n}$, ${\tt Coh}_{d,n}$, and ${\tt IU}_{d,n}$ represent the DL non-coherent interference, coherent interference, and the UE to UE CLI, respectively. These are expressed as
\begin{subequations}
	\begin{align}
	&\textstyle{{\tt NCoh}_{d,n}=N\sum\nolimits_{q\in\mathcal{U}_{d}}\sum\nolimits_{j\in\mathcal{A}_{d}}\mathcal{E}_{d,j}\kappa_{jq}^2\beta_{jn}\alpha_{jq}^2}\label{eq:dl_ncoh}
	\\&\textstyle{{\tt Coh}_{d,n}=N^{2}\sum\nolimits_{\substack{q\in \mathcal{I}_{p}\backslash n }}\big(\sum\nolimits_{j\in\mathcal{A}_{d}}\sqrt{\mathcal{E}_{d,j}}\kappa_{jq}\alpha_{jq}^2\textstyle{\sqrt{\frac{\mathcal{E}_{p,n}}{\mathcal{E}_{p,q}}}\frac{\beta_{jn}}{\beta_{jq}}}\big)^{2}}\label{eq:dl_coh}
	\\&\textstyle{{\tt IU}_{d,n}=\sum\nolimits_{\textcolor{black}{k\in\mathcal{U}_{u}}}\mathcal{E}_{u,n}\epsilon_{nk}}\label{eq:dl_UEUE}
	.\end{align}
\end{subequations}	
\end{thm}    
\begin{proof}
	See Appendix~\ref{appen: rate_dllink_thm}.
\end{proof}

Now, the overall sum UL-DL SE of the system can be expressed as
  \begin{align}\label{eq:sum_rate}
 \mathcal{R}_\text{sum}=\frac{\tau-\tau_{p}}{\tau}\big[\textstyle{\sum\nolimits_{k\in\mathcal{U}_{u}}}\mathcal{R}_{u,k}+\textstyle{\sum\nolimits_{n\in\mathcal{U}_{d}}}\mathcal{R}_{d,n}\big].
 \end{align}

We note that, from Theorems~\ref{thm: rate_UL_thm} and~\ref{thm:rate_final_dllink}, the gain and various interference terms involved in $\mathcal{R}_{u,k}$ and $\mathcal{R}_{d,k}$ are dependent on $\mathcal{{A}}_{u}$ and $\mathcal{{A}}_{d}$. 
Therefore, we obtain different values of $\mathcal{R}_{\text{sum}}$ for different choices of $\mathcal{{A}}_{u}$ and $\mathcal{{A}}_{d}$. To characterize this dependence, from this point onward, we write the achievable sum UL-DL SE as $\mathcal{R}_{\text{sum}}(\mathcal{A}_{x})$, where $\mathcal{A}_{x}\triangleq (\mathcal{A}_{u},\mathcal{A}_{d})$. Note that, as $\mathcal{U}_u$ and $\mathcal{U}_{d}$ are given, we omit their dependence on $\mathcal{R}_{\text{sum}}$. 
Now, the brute-force approach of listing out all the $2^{|\mathcal{{A}}_{u}\cup \mathcal{{A}}_{d}|}$ possible AP schedules and computing their achievable sum UL-DL SE using \eqref{eq:sum_rate} makes the complexity of finding an optimal AP schedule exponential in the number of APs. We present a low complexity solution in the next section.

\section{Sum Rate Optimization}\label{sec:Submodular}
We recall that the problem of finding the optimal AP schedule, namely, determining which APs should operate in the UL and which APs should operate in the DL, based on the local data demands from the UEs, is a combinatorially complex optimization problem. In this section, we circumvent this by developing a greedy AP-scheduling scheme based on sub-modularity. At each step of the procedure, we select which AP to schedule and whether the scheduled AP should operate in the UL or DL mode, such that the incremental gain in $\mathcal{R}_{\text{sum}}$ is maximized. This process is repeated until the last AP is scheduled, thereby solving the problem in polynomial time. Such a greedy approach to SE maximization has been previously proposed in the antenna selection literature, based on the monotonicity of the cost function~\cite{Han}. However, to provide concrete guarantees on the performance of the greedy search, we need to show that the cost function is a sub-modular set function of the scheduled APs. In this case, the greedy algorithm is guaranteed to yield a solution that achieves at least $(1-1/e)$-fraction of the optimal value of the cost function. For the sake of completeness, we formally define the monotonicity and sub-modularity properties as follows.
 \begin{defn}\cite{sub_mod}\label{defn:sub_mod}
	 Let $\mathcal{S}$ be a finite set, and let $2^\mathcal{S}$ denote its power set. A function $f : 2^{\mathcal{S}}\rightarrow\mathbb{R}$, with $f(\emptyset)=0$, is said to be \textbf{monotone nondecreasing} if for every $\mathcal{A}\subseteq
	\mathcal{B}\subseteq\mathcal{S}$,  $f(\mathcal{A})\leq f(\mathcal{B})$, and is said to be \textbf{sub-modular} if for every 
	$\{j\}\in\mathcal{S\backslash B}$,  $f(\mathcal{A}\cup\{j\}) -f(\mathcal{A})\geq f(\mathcal{B}\cup\{j\})-f(\mathcal{B}).$
\end{defn}

We first focus on the monotonicity of the sum UL-DL SE.
Let $\mathcal{A}$ be the indices of the APs in the network, where each AP (i.e., each index) can be  scheduled either in the UL or DL. Further, let $\mathcal{A}_s=\mathcal{A}_u\cup\mathcal{A}_d$ denote the index set of the APs that have been previously scheduled, and $\mathcal{A}_s'=\mathcal{A}\backslash\mathcal{A}_s$ be the index set of unscheduled APs. We need to show that adding an element from $\mathcal{A}_s'$ to $\mathcal{A}_s$ does not decrease $\mathcal{R}_{\text{sum}}$. 
Now, for any AP $m\in\mathcal{A}_s'$, let $\mathcal{A}_t \triangleq \mathcal{A}_s\cup\{m\}$. We note that when the $m$th AP is added to the set of UL APs, $\mathcal{A}_u$, it does not introduce any new interference, and hence the sum rate can only improve.  However, if the $m$th AP is added to $\mathcal{A}_d$, then it has the option to transmit with zero power. If it chooses to transmit at zero power, it is as if the AP was not added at all, so the sum rate obtained is the same as that obtained without it. However, if the AP optimally chooses a nonzero transmit power in order to maximize the sum rate, the sum rate can be potentially improved. Hence, the sum rate with the new AP added can only be  greater than or equal to the sum rate obtained without the AP, and $\mathcal{R}_\text{sum}(\mathcal{A}_{s})\leq \mathcal{R}_\text{sum}(\mathcal{A}_{t})$ with $\mathcal{A}_{s}\subseteq \mathcal{A}_{t}$. This shows that the sum rate is a monotone nondecreasing set function.

We now focus on the proof of sub-modularity. First, we observe that due to pilot contamination and CLIs, $\mathcal{R}_{\text{sum}}$ is a non-separable function of the scheduled AP sets $\mathcal{A}_{u}$ and $\mathcal{A}_{d}$. For example, if  the $\{j\}$th AP, $j\notin\mathcal{A}_{u}\cup\mathcal{A}_{d}$, is scheduled in the UL mode, we can write the gain and the coherent interference terms in~\eqref{eq:rate_final_uplink} and~\eqref{eq:Coh}  as
\begin{subequations}
	\begin{align}
		&\textstyle{\big(\sum\nolimits_{m\in\mathcal{A}_{u}\cup\{j\}}\alpha_{mk}^2\big)^{2}=\big(\sum\nolimits_{m\in\mathcal{A}_{u}}\alpha_{mk}^2\big)^{2}+\alpha_{jk}^{4}+2\alpha_{jk}^2\sum\nolimits_{m\in\mathcal{A}_{u}}\alpha_{mk}^2},\label{eq:Gain_NL}	
		\\&
		\textstyle{\sum\nolimits_{n\in \mathcal{I}_{p}\backslash k}\mathcal{E}_{u,n}\big(\sum\nolimits_{m\in{\mathcal{A}}_{u}\cup\{j\}}\alpha_{mk}^2\sqrt{\frac{\mathcal{E}_{p,n}}{\mathcal{E}_{p,k}}}\frac{\beta_{mn}}{\beta_{mk}}\big)^{2}=\sum\nolimits_{n\in \mathcal{I}_{p}\backslash k}\mathcal{E}_{u,n}\big[\big(\sum\nolimits_{m\in{\mathcal{A}}_{u}}\alpha_{mk}^2\sqrt{\frac{\mathcal{E}_{p,n}}{\mathcal{E}_{p,k}}}\frac{\beta_{mn}}{\beta_{mk}}\big)^{2}}\notag\\&\hspace*{4cm}+\textstyle{\alpha_{jk}^4{\frac{\mathcal{E}_{p,n}}{\mathcal{E}_{p,k}}}\frac{\beta_{jn}^2}{\beta_{jk}^2}+2\alpha_{jk}^2\sqrt{\frac{\mathcal{E}_{p,n}}{\mathcal{E}_{p,k}}}\frac{\beta_{jn}}{\beta_{jk}}\sum\nolimits_{m\in{\mathcal{A}}_{u}}\alpha_{mk}^2\sqrt{\frac{\mathcal{E}_{p,n}}{\mathcal{E}_{p,k}}}\frac{\beta_{mn}}{\beta_{mk}}\big]\label{eq:CInt_NL}},
	\end{align}
\end{subequations}
 respectively. We note that in~\eqref{eq:Gain_NL} and~\eqref{eq:CInt_NL}, the first two terms in the right hand side correspond to the gain and coherent interferences due to set $\mathcal{A}_{u}$ and scheduled $\{j\}$th UL AP, respectively. However, due to the nonlinearity and the cross terms, the UL SINR cannot be written as a separable function of the set of scheduled APs. Thus, $\eta_{u,k}(\mathcal{A}_{u}\cup\{j\})\neq\eta_{u,k}(\mathcal{A}_{u})+\eta_{u,k}(\{j\})$. Similar observations hold in DL. Furthermore, in our system, the UL SINRs and the DL SINRs are coupled with the DL transmitted signals via the AP-to-AP CLI and UL transmitted signals via UE-to-UE CLI, respectively, which makes the SINRs dependent on the power control coefficients. Therefore, our problem becomes challenging compared to previous works in antenna selection and UE scheduling literature which have considered either linear cost functions with respect to the maximization sets~\cite{Asymp} or perfect CSI at the APs~\cite{RV}. 
 
In several studies, the authors rely on reasonable system approximations, such as, high SNR~\cite{TWC_beam}, or  the
SE under asymptotic antenna density~\cite{Asymp}, which lead to tractable analytical expressions. Such approximate cost function based analysis is known as \emph{sub-modular relaxation}~\cite{TWC_beam}. In this work, we note that as the number of antennas at each AP, $N$, goes to infinity, the non-coherent interferences becomes negligible compared to the gain and coherent interferences
as observed in Theorem~\ref{thm: rate_UL_thm} and Theorem~\ref{thm:rate_final_dllink}. Also, in a CF-system, the CPU can potentially cancel the AP-AP CLI with the global knowledge of the DL data streams. Therefore, to make the analysis tractable, we bound both the UL and DL rates and formulate an equivalent optimization problem based on the product SINR. 
Note that, as $N\rightarrow\infty$, we can show that
	\begin{align}
		&\textstyle{\mathcal{R}_{u,k}\geq\log_{2}\frac{\mathcal{E}_{u,k}(\sum_{m\in\mathcal{A}_{u}}\alpha_{mk}^2)^{2}}{\sum\limits_{n\in \mathcal{I}_{p}\backslash k}\!\!\!\mathcal{E}_{u,n}(\sum\limits_{m\in\mathcal{A}_{u}} \!\!\!\! \alpha_{mk}^2\sqrt{\frac{\mathcal{E}_{p,n}}{\mathcal{E}_{p,k}}}\frac{\beta_{mn}}{\beta_{mk}})^{2}}\geq\log_{2}\frac{(\sqrt{\mathcal{E}_{u,k}}{\sum_{m\in\mathcal{A}_{u}}}\alpha_{mk}^2)^{2}}{(\sum\limits_{n\in \mathcal{I}_{p}\backslash k}\sqrt{\mathcal{E}_{u,n}}\sum\limits_{m\in\mathcal{A}_{u}}\alpha_{mk}^2\sqrt{\frac{\mathcal{E}_{p,n}}{\mathcal{E}_{p,k}}}\frac{\beta_{mn}}{\beta_{mk}})^{2}}},\label{eq:approx_high_SINR_ul} \\
		&\textstyle{\mathcal{R}_{d,n}\geq\log_{2}\frac{(\sum_{j\in\mathcal{A}_{d}}\kappa_{jn}{\sqrt{\mathcal{E}_{j,n}}\alpha_{jn}^2})^{2}}{\sum\limits_{\substack{q\in \mathcal{I}_{p}\backslash n }}\!\!(\sum\limits_{j\in\mathcal{A}_{d}}\!\!\!\!\sqrt{\mathcal{E}_{d,j}}\kappa_{jq}\alpha_{jq}^2\sqrt{\frac{\mathcal{E}_{p,n}}{\mathcal{E}_{p,q}}}\frac{\beta_{jn}}{\beta_{jq}})^{2}}\geq\log_{2}\frac{(\sum\nolimits_{j\in\mathcal{A}_{d}}\kappa_{jn}{\sqrt{\mathcal{E}_{j,n}}}{\alpha_{jn}^2})^{2}}{(\sum\limits_{\substack{q\in \mathcal{I}_{p}\backslash n }}\sum\limits_{j\in\mathcal{A}_{d}}\!\!\!\sqrt{\mathcal{E}_{d,j}}\kappa_{jq}\alpha_{jq}^2\sqrt{\frac{\mathcal{E}_{p,n}}{\mathcal{E}_{p,q}}}\frac{\beta_{jn}}{\beta_{jq}})^{2}}}\label{eq:approx_high_SINR_dl} .
		\end{align}
The latter lower bounds in~\eqref{eq:approx_high_SINR_ul} and~\eqref{eq:approx_high_SINR_dl} follow as we have only added more interference terms in the denominators.
Let $\mathcal{R}'_{u,k}=\log_{2}\frac{(\sqrt{\mathcal{E}_{u,k}}\sum_{m\in\mathcal{A}_{u}}\alpha_{mk}^2)^{2}}{(\sum\nolimits_{n\in \mathcal{I}_{p}\backslash k}\sqrt{\mathcal{E}_{u,n}}\sum\nolimits_{m\in\mathcal{A}_{u}}\alpha_{mk}^2\sqrt{\frac{\mathcal{E}_{p,n}}{\mathcal{E}_{p,k}}}\frac{\beta_{mn}}{\beta_{mk}})^{2}}=2\log_{2}\frac{\sum_{m\in\mathcal{A}_{u}}{\tt G}_{u,mk}}{\sum_{m\in\mathcal{A}_{u}}{\tt I}_{u,mk}}$ and $\mathcal{R}'_{d,n}=\log_{2}\frac{(\sum_{j\in\mathcal{A}_{d}}\kappa_{jn}\sqrt{\mathcal{E}_{j,n}}{\alpha_{jn}^2})^{2}}{(\sum\nolimits_{\substack{q\in \mathcal{I}_{p}\backslash n }}\sum_{j\in\mathcal{A}_{d}}\sqrt{\mathcal{E}_{d,j}}\kappa_{jq}\alpha_{jq}^2\sqrt{\frac{\mathcal{E}_{p,n}}{\mathcal{E}_{p,q}}}\frac{\beta_{jn}}{\beta_{jq}})^{2}}=2\log_{2}\frac{\sum_{j\in\mathcal{A}_{d}}{\tt G}_{d,jn}}{\sum_{j\in\mathcal{A}_{d}}{\tt I}_{d,jn}}$,
with the respective terms being defined as $\textstyle{{\tt G}_{u,mk}\triangleq \sqrt{\mathcal{E}_{u,k}}\alpha_{mk}^2}$, $\textstyle{{\tt I}_{u,mk}\triangleq\sum\nolimits_{n\in \mathcal{I}_{p}\backslash k}\sqrt{\mathcal{E}_{u,n}}\alpha_{mk}^2\sqrt{\frac{\mathcal{E}_{p,n}}{\mathcal{E}_{p,k}}}\frac{\beta_{mn}}{\beta_{mk}}}$, $\textstyle{{\tt G}_{d,jn}\triangleq \kappa_{jn}\sqrt{\mathcal{E}_{j,n}}{\alpha_{jn}^2}}$, and $\textstyle{{\tt I}_{d,jn}\triangleq\sum\nolimits_{\substack{q\in \mathcal{I}_{p}\backslash n }}\sqrt{\mathcal{E}_{d,j}}\kappa_{jq}\alpha_{jq}^2\sqrt{\frac{\mathcal{E}_{p,n}}{\mathcal{E}_{p,q}}}\frac{\beta_{jn}}{\beta_{jq}}}$. Now, given the set of APs $\mathcal{A}_{s}$, our problem is to optimally decide the partition $\mathcal{A}_{u}$ and $\mathcal{A}_{d}$ such that the sum UL-DL SE, i.e. $\mathcal{R}'_\text{sum}\triangleq[
\sum\nolimits_{k\in\mathcal{U}_{u}}\mathcal{R}'_{u,k}+\sum\nolimits_{n\in\mathcal{U}_{d}}\mathcal{R}'_{d,n}]$, is maximized. For notational simplicity, we rewrite our problem as follows
\begin{align}\label{eq:simplified_cost}
\textstyle{\max_{\mathcal{A}_{s}}\mathcal{R}'_{\text{sum}}=\max_{\mathcal{A}_{s}} \sum_{k=1}^{K}2\log_{2}\frac{\sum_{m\in\mathcal{A}_{s}}{\tt G}_{mk}(\mathcal{A}_{s})}{\sum_{m\in\mathcal{A}_{s}}{\tt I}_{mk}(\mathcal{A}_{s})}{(a)\atop\equiv }\max_{\mathcal{A}_{s}} {\prod_{k=1}^{K}\frac{\sum_{m\in\mathcal{A}_{s}}{\tt G}_{mk}(\mathcal{A}_{s})}{\sum_{m\in\mathcal{A}_{s}}{\tt I}_{mk}(\mathcal{A}_{s})}}},
\end{align}
where the $k$th UE can be either UL or DL and $m$th AP is either scheduled in UL or in the DL. Here, we explicitly write the gain and interferences as a function of $\mathcal{A}_{s}$. The equivalence in $(a)$ follows from the monotonicity of $\log_{2}(.)$.
\begin{thm}\label{thm:submod}
	The product SINR, $\textstyle{f_{mk}(\mathcal{A}_{x})=\prod\limits_{k=1}^{K}\frac{\sum_{m\in\mathcal{A}_{x}}{\tt G}_{mk}(\mathcal{A}_{x})}{\sum_{m\in\mathcal{A}_{x}}{\tt I}_{mk}(\mathcal{A}_{x})}}$, is a sub-modular function of the number of scheduled APs in the system. That is, if $\mathcal{A}_{s}$ and $\mathcal{A}_{t}$ are  index sets of active APs, with $\mathcal{A}_{s}\subseteq \mathcal{A}_{t}$, and if $\{j\}\notin \mathcal{A}_{t}$, then $	\textstyle{f_{mk}(\mathcal{A}_{s}\cup\{j\})-f_{mk}(\mathcal{A}_{s})\geq f_{mk}(\mathcal{A}_{t}\cup\{j\})-f_{mk}(\mathcal{A}_{t})}.$
\end{thm}
\begin{proof}
	See Appendix~\ref{proof:submod_thm}.
\end{proof}

\begin{algorithm}[!t]
	\SetAlgoLined
	\textbf{Inputs}: 	 $\mathcal{A}$: the set of all AP indices\;
	\textbf{Initialize} $\mathcal{A}_u=\mathcal{A}_d=\emptyset$,
	 $\mathcal{A}_s=\mathcal{A}_u\cup\mathcal{A}_d$\;
	
	\While{ $\mathcal{A}_s'\neq\emptyset$ }
		{
		$i^{\star}_{u}=\arg \max\limits_{i\in \mathcal{A}_s'} \mathcal{R}'_\text{sum} (\mathcal{{A}}_u\cup\{i\})$\;
		$i^{\star}_{d}=\arg \max\limits_{i\in \mathcal{A}_s'} \mathcal{R}'_\text{sum} (\mathcal{{A}}_d\cup\{i\})$\;				
		\If{$\mathcal{R}'_\text{sum} (\mathcal{{A}}_u\cup\{i_u^{\star}\})\ge\mathcal{R}'_\text{sum} (\mathcal{{A}}_d\cup\{i_d^{\star}\})$}{\textbf{Update}~~$\mathcal{A}_{u}=\mathcal{A}_{u}\cup\{i_u^{\star}\}$\; 
		\Else{\textbf{Update}  $\mathcal{A}_d=\mathcal{A}_d\cup \{i_d^{\star}\}\;$}
	}

		$\mathcal{A}_s=\mathcal{A}_u\cup\mathcal{A}_d$\;
		
	}
    \textbf{Return}~~{$\mathcal{{A}}_u$ and $\mathcal{A}_d$}\;
	\caption{Greedy algorithm for AP scheduling}\label{algo:GA_submod}
\end{algorithm}

Now, exploiting the sub-modularity of $\mathcal{R}'_\text{sum}$, we can develop a greedy algorithm for scheduling the APs as detailed in Algorithm~\ref{algo:GA_submod}.
It follows that $\mathcal{R}'_\text{sum}(\mathcal{\dot{A}})\geq \left(1-\frac{1}{e}\right)\mathcal{R}'_\text{sum}(\mathcal{A}^{\star})$, where $\mathcal{A}^{\star}$ is the index set containing the optimal AP configuration that maximizes the cost function $\mathcal{R}'_\text{sum}$, and $\mathcal{\dot{A}}$ denotes the AP configuration returned by  Algorithm~\ref{algo:GA_submod}. 
We validate the effectiveness of the greedy algorithm via numerical simulations in Fig.~\ref{fig: valid_submod_search}. 
The simulation parameters are detailed in Sec.~\ref{sec:simulation}. For the brute force based search, we have considered our original cost function $\mathcal{R}_\text{sum}$ as expressed in~\eqref{eq:sum_rate} over all possible AP-schedules to find the optimal SE. We also use the AP schedule generated by Algorithm~\ref{algo:GA_submod} and evaluate $\mathcal{R}_\text{sum}$ using~\eqref{eq:sum_rate}.  We observe that the sum UL-DL SE obtained via exhaustive search over all $2^M$ UL-DL AP-configurations and considering the effects of CLIs matches closely with that obtained via Algorithm~\ref{algo:GA_submod} based on sub-modularity of the lower-bounded cost function. 

In Fig.~\ref{fig:theory_vs_simulation}, we plot the $90\%$-likely sum UL-DL SE vs. UL and DL data SNR to validate the theoretical expressions of SE derived in Theorems~\ref{thm: rate_UL_thm} and~\ref{thm:rate_final_dllink}. For the simulation, we consider $10,000$ Monte Carlo channel instantiations and UE locations; the other parameters used can be found in Sec.~\ref{sec:simulation}. The theoretical curve is obtained by averaging the $90\%$ likely sum SE obtained from \eqref{eq:sum_rate} over the UE locations. The simulation corroborates well with our derived results, verifying the accuracy of the expression for $\mathcal{R}_{\text{sum}}$  derived above.

\begin{figure}[!t]
\begin{subfigure}[b]{0.49\textwidth}
	\centering
	\includegraphics[width=0.8\textwidth]{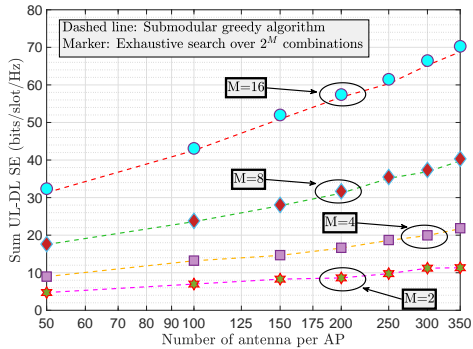}
\caption{Sum UL-DL SE~(bits/slot/Hz) vs number of AP-antennas for different numbers of APs. This plot shows the effectiveness of the sub-modular algorithm.}\label{fig: valid_submod_search}
\end{subfigure}
\hfill
\begin{subfigure}[b]{0.49\textwidth}
	\centering
	\includegraphics[width=0.8\textwidth]{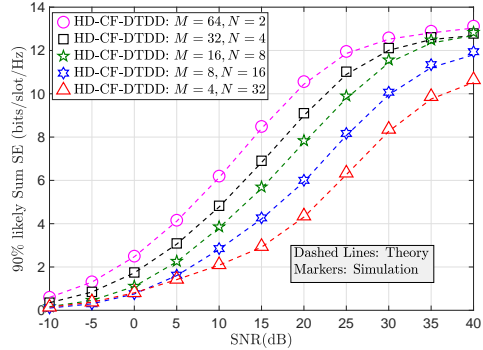}
	\caption{The $90\%$-likely sum UL-DL SE vs. data SNR with $K=100$. This figure validates the derived theoretical expressions of the sum SE with Monte Carlo simulations.}\label{fig:theory_vs_simulation}
\end{subfigure}
\caption{Validation of sum-modular algorithm and the derived expressions of sum UL-DL SE under MRC/MFP.}
\vspace*{-0.8cm}
\end{figure}

\section{Performance Analysis: MMSE \& RZF}
\textcolor{black}{It is known that the performance of CF-mMIMO can be  improved with centralized MMSE combining in the UL and RZF precoding in the DL~\cite{making_cf_emil,net_work}. In this section we briefly analyze the performance of our system model under these combining and precoding schemes.}

\textcolor{black}{Let $\mathcal{A}_{u}(m)$ and $\mathcal{U}_{u}(k)$ denote the $m$th UL AP and the $k$th UL UE  in $\mathcal{A}_{u}$ and $\mathcal{U}_{u}$, respectively. Let $\mathcal{A}_{d}(j)$ and $\mathcal{U}_{d}(n)$ denote the $j$th DL AP and the $n$th DL UE  in $\mathcal{A}_{d}$ and $\mathcal{U}_{d}$, respectively. Let $\hat{\mathbf{f}}_{u,k}\in\mathbb{C}^{N|\mathcal{A}_{u}|}$ denote the estimated channel matrix of the $k$th UL UE to all the UL APs, i.e.,  $\hat{\mathbf{f}}_{u,k}=\left[\hat{\mathbf{f}}_{\mathcal{A}_{u}(1)k}^T,\ldots,\hat{\mathbf{f}}_{\mathcal{A}_{u}(|\mathcal{A}_{u}|)k}^T\right]^T, \forall k\in\mathcal{U}_{u}$, and let the estimated UL channel matrix available at the CPU be denoted by $\hat{\mathbf{F}}_{u}\triangleq\left[\hat{\mathbf{f}}_{u,\mathcal{U}_{u}(1)},\ldots,\hat{\mathbf{f}}_{u,\mathcal{U}_{u}(|\mathcal{U}_{u}|)}\right]\in\mathbb{C}^{N|\mathcal{A}_{u}|\times|\mathcal{U}_{u}|}$. Similarly, we can express the estimated channel of the DL UEs as $\hat{\mathbf{F}}_{d}\triangleq\left[\hat{\mathbf{f}}_{d,\mathcal{U}_{d}(1)},\ldots,\hat{\mathbf{f}}_{d,\mathcal{U}_{d}(|\mathcal{U}_{d}|)}\right]\in\mathbb{C}^{N|\mathcal{A}_{d}|\times|\mathcal{U}_{d}|}$ with
$\hat{\mathbf{f}}_{d,n}=\left[\hat{\mathbf{f}}_{\mathcal{A}_{d}(1)n}^T,\ldots,\hat{\mathbf{f}}_{\mathcal{A}_{d}(|\mathcal{A}_{d}|)n}^T\right]^T\in\mathbb{C}^{N|\mathcal{A}_{d}|}, \forall n\in\mathcal{U}_{d}$. Now, the concatenated UL  signal received at the CPU can be expressed as
\begin{align}
	\mathbf{y}_{u}=\sum\nolimits_{k\in\mathcal{U}_{u}}\sqrt{\mathcal{E}_{u,k}}\left(\hat{\mathbf{f}}_{u,k}+\tilde{\mathbf{f}}_{u,k}\right)s_{u,k}+\sum\nolimits_{n\in\mathcal{U}_{d}}\mathbf{G}\mathbf{p}_{n}s_{d,n}+\mathbf{w}_{ul},\tag{18}
\end{align}
where $\mathbf{G}\in\mathbb{C}^{N|\mathcal{A}_{u}|\times N|\mathcal{A}_{d}|}$ denotes the residual interference channel between DL APs and UL APs, and  $\mathbf{p}_{n}\in\mathbb{C}^{N|\mathcal{A}_{d}|}$ is the $n$th column of the DL precoder $\mathbf{P}=\left[\mathbf{p}_{\mathcal{U}_{d}(1)},\ldots,\mathbf{p}_{\mathcal{U}_{d}(|\mathcal{U}_{d}|)}\right]\in\mathbb{C}^{N|\mathcal{A}_{d}|\times |\mathcal{U}_{d}|}$.  With a slight abuse of notation, let  $\mathbf{P}_{j}\in\mathbb{C}^{N\times|\mathcal{U}_{d}|}$ denote the precoding matrix for the $j$th DL AP, and let $\mathcal{E}_{d,j}$ denote the power budget per antenna at the $j$th DL AP, so that the power constraint becomes $\text{tr}(\mathbf{P}_{j}\mathbf{P}_{j}^{H})\leq N\mathcal{E}_{d,j}$~\cite{net_work}. Finally, $\mathbf{w}_{ul}\sim\mathcal{CN}(\mathbf{0},\mathbf{I}_{N|\mathcal{A}_{u}|})$ is the additive noise. Then, $\mathbf{V}=\mathbf{Q}_{u}^{-1}\hat{\mathbf{F}}_{u}\in\mathbb{C}^{N|\mathcal{A}_{u}|\times |\mathcal{U}_{u}|}$ is the joint MMSE combiner, with $	\mathbf{Q}_{u}=\left(\sum_{k\in\mathcal{U}_{u}}\mathcal{E}_{u,k}\hat{\mathbf{f}}_{u,k}\hat{\mathbf{f}}_{u,k}^{H}+\mathbf{R}_{u}+N_{0}\mathbf{I}_{N|\mathcal{A}_{u}|}\right)$, where
\begin{align}
	\mathbf{R}_{u}=\left(\sum\nolimits_{k\in\mathcal{U}_{u}}\mathcal{E}_{u,k}\mathbb{E}\left[\tilde{\mathbf{f}}_{u,k}\tilde{\mathbf{f}}_{u,k}^{H}\right]+\sum\nolimits_{i\in\mathcal{U}_{d}}\mathbb{E}\left[\mathbf{G}\mathbf{p}_{n}\mathbf{p}_{n}^H\mathbf{G}^H\right]\right).\tag{19}
\end{align}
Then, the UL sum SE becomes~\cite{making_cf_emil} $\mathcal{R}_{u}=\sum_{k\in\mathcal{U}_{u}}\mathbb{E}\log_{2}\left(1+\eta_{u,k}\right)$, where
\begin{equation}\label{eq:UL_MMSE_SINR}
	\eta_{u,k}=\frac{\mathcal{E}_{u,k}\left|\hat{\mathbf{f}}_{u,k}^H\mathbf{Q}_{u}^{-1}\hat{\mathbf{f}}_{u,k}\right|^2}{\sum\nolimits_{k'\in\mathcal{U}_{u}\setminus k}\mathcal{E}_{u,k'}\left|\hat{\mathbf{f}}_{u,k}^{H}\mathbf{Q}_{u}^{-1}\hat{\mathbf{f}}_{u,k'}\right|^2+\hat{\mathbf{f}}_{u,k}^H\mathbf{Q}_{u}^{-1}\left(\mathbf{R}_{u}+N_{0}\mathbf{I}_{N|\mathcal{A}_{u}|}\right)\mathbf{Q}_{u}^{-1}\hat{\mathbf{f}}_{u,k}},\tag{20}
\end{equation}
with the expectation being taken over the channel realizations. The  MMSE combiner presented here maximizes the $k$th UL UE's instantaneous SINR~\cite{massivemimobook, making_cf_emil}.} 

\textcolor{black}{In the DL, the RZF precoder is a commonly used linear  precoding scheme to control inter-UE interference~\cite{net_work}. It is designed as $\mathbf{P}=\kappa\mathbf{Q}_{d}^{-1}\hat{\mathbf{F}}_{d}$, where $\mathbf{Q}_{d}=\left(\hat{\mathbf{F}}_{d}\hat{\mathbf{F}}_{d}^{H}+\xi\mathbf{I}_{N|\mathcal{A}_{d}|}\right)$, $\kappa$ is the power normalization factor, and $\xi>0$ is a regularization parameter~\cite{Sir_RZF, RZF_vector}.   The DL sum SE can be increased by appropriately selecting $\xi$~\cite{Sir_RZF}, and the DL power control parameter $\kappa$ is  evaluated at the CPU based on the estimated channel statistics. Considering an equal power budget at each DL AP, i.e. $\mathcal{E}_{d,j}=\mathcal{E}_{d}$, it is easy to show $\kappa_{j}^2=N\mathcal{E}_{d}/\text{tr}(\mathbf{P}_{j}\mathbf{P}_{j}^H)$ satisfies the DL power constraint. We set $\kappa^2=\min_{j}\kappa_{j}^2$, for all $j\in\mathcal{A}_{d}$, an approach previously used in~\cite{RZF}. We consider that the DL UEs know the mean of the precoded signal, and therefore, applying the use-and-then forget bound, we can write the DL SE as $\mathcal{R}_{d}=\sum_{n\in\mathcal{U}_{d}}\log_{2}\left(1+\eta_{d,n}\right)$, with 
\begin{align}\label{eq:DL_RZF_SINR}
\eta_{d,n}=\dfrac{\kappa^2\left|\mathbb{E}\left[{\mathbf{f}}_{n}^H\mathbf{Q}_{d}^{-1}\hat{\mathbf{f}}_{n}\right]\right|^2}{\kappa^2\sum\nolimits_{n'\in\mathcal{U}_{d}\backslash n}\mathbb{E}\left[\left|{\mathbf{f}}_{n}^H\mathbf{Q}_{d}^{-1}\hat{\mathbf{f}}_{n'}\right|^2\right]+{\tt var}\left({\mathbf{f}}_{n}^H\mathbf{Q}_{d}^{-1}\hat{\mathbf{f}}_{n}\right)+\sum\nolimits_{k\in\mathcal{U}_{u}}\mathcal{E}_{u,k}\mathbb{E}\big|\mathtt{g}_{nk}\big|^2+N_{0}}\tag{21}
\end{align}
being the DL SINR of the $n$th UE, where the expectations are taken over the channel realizations.}

\textcolor{black}{With the above UL and DL SE expressions in hand, we can compare the performance of MRC/MFP based combiner/precoding and the MMSE-type combiner/precoder. The APs are scheduled according to  Algorithm~\ref{algo:GA_submod}, with the sum rate computed using the UL and DL SINRs evaluated according to~\eqref{eq:UL_MMSE_SINR} and ~\eqref{eq:DL_RZF_SINR}, respectively.} 
\textcolor{black}{In Fig.~\ref{fig:perfoamnce_comp_MMSE_MRC_RZF_MFP},  
we see that, with $(M=64, N=4)$, the $90\%$-sum UL-DL SE achieved via MMSE/RZF is double the sum UL-DL SE achieved via MRC/MFP under similar settings. This shows the interference suppression capability of MMSE-based combiner and precoder, as well as the benefits of the centralized MMSE-processing scheme. However, the complexity of these schemes increase significantly with system dimension, i.e., number of UEs and number of APs. Also, when we increase the number of APs from $8$ to $64$, we observe a substantial performance improvement irrespective of the processing scheme. There are two contributing factors to this improvement: First, as we increase the number of APs, the flexibility to schedule the APs either in UL or in DL mode also increases, and therefore, the sum UL-DL SE improves considerably. Second, with more APs, the probability that an UE finds an AP (or APs) in its proximity also increases, and which in turn improves the rate achieved by that UE, leading to an improvement in sum UL-DL SE. 
 }

\begin{figure}
\centering
\includegraphics[width=0.4\textwidth]{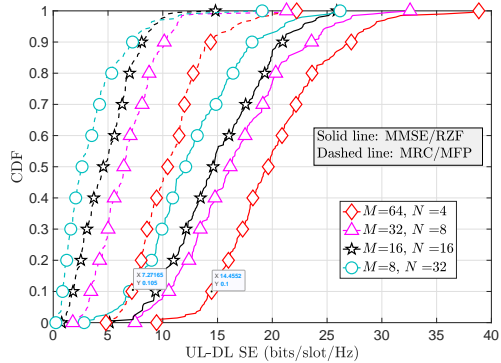}
\caption{\textcolor{black}{Performance comparison of MMSE-type precoder/combiner with MRC/MFP with $K=32$.
	}}\label{fig:perfoamnce_comp_MMSE_MRC_RZF_MFP}
\vspace*{-0.8cm}
\end{figure}

\section{Pilot Allocation}\label{sec: pilot_allocation}
The problem of pilot allocation in a CF-mMIMO system is fundamentally different from similar problems encountered in the cellular version. In cellular mMIMO systems, we can avoid intra-cell pilot contamination by allotting orthogonal pilots to all the UEs within each cell~\cite{my_MMSE_SPAWC}. Also, in cellular systems, only the serving BS estimates the channel for each UE. However, in a CF system, it is not just the nearest AP that is interested in estimating the channel of a given UE; all the (nearby) APs need to estimate the channel in order to correctly combine the signals from all the UEs. Moreover, allotting orthogonal pilots to all UEs in the entire geographical zone could lead to a very high pilot overhead as the length of orthogonal pilots scales linearly with the total number of UEs~\cite{Nayebi,random_vs_structured,Clustered_CF,Pilot_sum_rate}. On the other hand, reusing a set of orthogonal pilots can result in pilot-sharing UEs in close proximity, leading to high pilot contamination. Hence, there is a need to revisit the problem of pilot allocation across UEs in the context of CF systems.  

\textcolor{black}{The authors in~\cite{Clustered_CF,random_vs_structured} present a joint AP-UE clustering based pilot allocation algorithm that maximizes the sum SE. However, we consider the canonical CF-architecture where every UE can potentially be served by all the APs~\cite{cell_small,Nayebi}, and hence, the clustering based algorithms cannot be directly applied to our model. In~\cite{Clustered_CF,random_vs_structured}, TDD is considered, where, after clustering, all clusters operate either in UL or in DL at a given point in time. However, in a DTDD based system, the randomly distributed UEs have different UL/DL data demands. Due to this, jointly clustering the AP-UEs in DTDD is a more involved problem. The authors in~\cite{Pilot_sum_rate} develop a UE-centric dynamic clustering based pilot allocation scheme to maximize the UL SE. The solution requires the knowledge of the SINR at each UE under all possible pilot allocations to obtain the optimal pilot-UE pair, which in turn significantly increases the signaling overhead when the number of UEs is large. 
To summarize, the existing approaches to pilot allocation require significant signaling overhead, and also cannot be directly applied to the settings in our problem.
The key feature of our pilot allocation strategy is that by employing pilots before the AP scheduling and data transmission phase, we 
can decouple the problems of pilot allocation and AP-scheduling, in a DTDD based system, which makes our solution attractive for  implementation.}

We now develop an iterative algorithm to allocate the pilots to different UEs. Recall from Sec.~\ref{sec: channel_estimation} that a UE with a good channel estimate at the $m$th AP will have a higher value of $\alpha_{mk}^2=c_{mk}\tau_{p}\mathcal{E}_{p,k}\beta_{mk}^2$, where
the term $c_{mk}$ accounts for pilot contamination. As the distance between $k$th and $n$th $(k,n\in \mathcal{I}_p)$ UEs decreases, the values of $\alpha_{mk}^2$ and $\alpha_{mn}^2$ decrease, resulting in worsening of the channel estimates for both UEs.
Hence, we first arbitrarily allocate pilots to all the $K$ UEs, and then we find the UE $k^{\star}$ with the least value of $\alpha_{mk}$ to its nearest AP,  that is,  $k^{\star}=\arg \min _{k}\alpha_{mk}$, where $m$ is the index of the AP closest to UE $k$. If $\boldsymbol{\phi}_{k^{\star}}$ is the associated pilot for this UE,  we reallocate a new pilot sequence to this UE from $\{\boldsymbol{\phi}_{1},\ldots,\boldsymbol{\phi}_{K}\}\backslash\boldsymbol{\phi}_{k^{\star}}$ so that $\alpha_{mk^{\star}}$ is maximized. We repeat this iterative process either up to a predetermined number of iterations, or if no other pilot sequence from  $\{\boldsymbol{\phi}_{1},\ldots,\boldsymbol{\phi}_{K}\}\backslash\boldsymbol{\phi}_{k^{\star}}$ improves  $\alpha_{mk^{\star}}$, or if $\alpha_{mk^{\star}}$ exceeds a certain threshold for all UEs. The overall recipe is presented in Algorithm~\ref{GA_pilot}.  

We illustrate the effectiveness of the proposed algorithm in Fig.~\ref{fig:algo_pilot_alloc_verify}. In the cell-based allocation scheme, since $\tau_p = 25$, we consider $4$ equal sized cells in the system, assign each UE to its nearest cell center, allot orthogonal pilots to the UEs within each cell. In case the number of UEs in any cell exceeds the pilot length, we set the pilot length to equal the maximum group size, thus maintaining orthogonality of the pilots within each cell. This reduces pilot contamination within each cluster of UEs, and therefore outperforms random pilot allocation. However, we see that pilot allocation according to Algorithm~\ref{GA_pilot} significantly improves the overall SE compared to both cell-based grouping and random pilot allocation schemes. 

\begin{figure}[!t]
	\centering
	\includegraphics[width=0.43\textwidth]{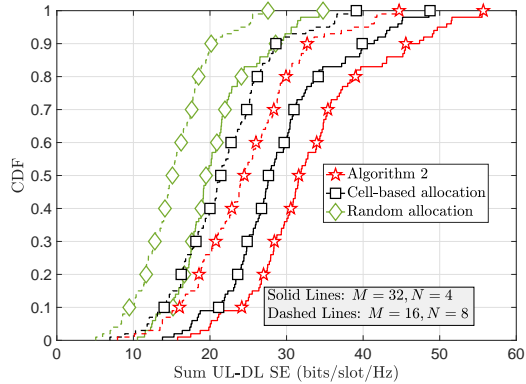}
	\caption{Cumulative distribution function~(CDF) of the achievable sum UL-DL SE under the pilot allocation obtained from Algorithm~\ref{GA_pilot}. Here, $K=100, \tau_{p}=25, N_{\text{iter}}=1000$, and pilot SNR $=20$~dB.}
	\label{fig:algo_pilot_alloc_verify}
	\vspace*{-0.4cm}
\end{figure}

\begin{algorithm}[!t]
	\SetAlgoLined
	\textbf{Initialize}: $\mathcal{I}_{p}$ for $1\le p \le \tau_p$,  Number of iterations=$N_{{iter}}$\;% Number of iterations=$N$\;
	\textbf{Calculate} $d_{mk}$ for all $k\in\mathcal{U}$, $m\in \mathcal{A}$\;
	\textbf{Define}: $m^{\star}_k=\arg\min\limits_{m} d_{mk}$\;
	\textbf{Calculate} $\alpha_{m_k^{\star} k}$ for all $k\in\mathcal{U}$\;
	\textbf{Initialize}: $\alpha_{{q}}=\max\limits_{k} \alpha_{m_k^{\star}k}$\;
	$\alpha_{\text{min}}=\min\limits_{k} \alpha_{m_k^{\star}k}$\; 
	\While{($\alpha_{\text{min}}<\alpha_{o}$)~\&\&~($\alpha_{\text{min}}<\alpha_{q}$)~\&\&~$({\tt i}\leq N_{\text{iter}})$}{
		$k^{\star}=\arg \min\limits_k \alpha_{mk}$\;
		$\alpha_q=\alpha_{\text{min}}$\;
		\For{$1\le p \le \tau_p$} 
		{
			$\mathcal{I}_p=\mathcal{I}_p \cup \{k^{\star}\} $ \;
			$a_p=\alpha_{m_k^{\star}k}$\;
	
	}
     $p^{\star}=\arg\max\limits_{p}a_p$\;
     Reallocate $k$ to $\mathcal{I}_{p^{\star}}$\;
     \textbf{Update:} $\alpha_{\text{min}}=\min\limits_{k} \alpha_{m_k^{\star}k}$ \;
     \textbf{Set:} ${\tt i}={\tt i}+1$\;
	}
	\caption{Iterative Pilot Allocation}\label{GA_pilot}
\end{algorithm}
\vspace*{-.3cm}
\section{Full-Duplex Multi-cell Systems}
\textcolor{black}{In this section, we briefly present the sum UL-DL SE achieved by an FD-enabled multi-cell mMIMO system, based on~\cite{Asymptotic},  to enable fair comparison with the DTDD based CF-mMIMO system.} 
We assume that each cell has one FD BS with $N_{t}$ transmit and $N_{r}$ receive antennas. To maintain the consistency with our previous analysis, we assume that the total number of UEs across all cells is same as the total number of UEs~($K$) in the CF system. Let $\mathcal{U}_{l,u}$ and $\mathcal{U}_{l,d}$ denote the index sets of HD UL and DL UEs within the $l$th cell, such that $\sum_{l=1}^{L}(|\mathcal{U}_{l,u}|+|\mathcal{U}_{l,d}|)=K$. We also assume that each FD BS can perfectly cancel out its self-interference. However, we do not assume any inter-BS cooperation for interference management.  Therefore, each BS experiences interference from neighboring cells. Let the UL channel from $k$th UE of $l$th cell to the $j$th BS be denoted by $\mathbf{f}_{u,jlk}=\sqrt{\beta_{jlk}}\mathbf{h}_{u,jlk}\in\mathbb{C}^{N_{r}\times 1}$, with $\beta_{jlk}$ being the slow fading component that includes the path loss, and $\mathbf{h}_{u,jlk}\sim\mathcal{CN}(0,\mathbf{I}_{N_r})$ being the fast fading component. Similarly, the DL channel from the $j$th BS to the $n$th DL UE of the $l$th cell can be modeled as $\mathbf{f}_{d,jln}=\sqrt{\beta_{jln}}\mathbf{h}_{d,jln}\in\mathbb{C}^{N_{t}\times 1}$. The channel matrix from the DL antenna array of the $j$th BS to the UL antenna array of the $l$th BS  is denoted by $\mathbf{T}_{jl}\in\mathbb{C}^{N_{r}\times N_{t}}$, with each element modeled as $\mathcal{CN}(0,\rho_{ij})$. 
We model the channel between the $k$th UL UE of the $l$th cell and the $n$th DL UE of $l'$th cell as $\mathtt{g}_{lk,l'n}\sim\mathcal{CN}(0,\epsilon_{lk,l'n})$.%, 
 In the channel estimation phase, we assume that all the UL and DL UEs synchronously transmit orthogonal pilots for channel estimation~\cite{Asymptotic,my_FD_SPAWC}. 
The UL and the DL estimated channels $\hat{\mathbf{f}}_{u,jlk}$ and $\hat{\mathbf{f}}_{d,jlk}$ of ${\mathbf{f}}_{u,jlk}$ and ${\mathbf{f}}_{d,jlk}$, respectively,  can be expressed as ${\mathbf{f}}_{x,jlk}={\hat{\mathbf{f}}}_{x,jlk}+{\tilde{\mathbf{f}}}_{x,jlk}$, $x\in{u,d}$,
with $\tilde{\mathbf{f}}_{x,jlk}$ being the estimation error vector, consisting of i.i.d. entries such that
$\tilde{\mathbf{f}}_{x,jlk}\sim\mathcal{CN}(0,(\beta_{jlk}-{\sigma^2_{jlk}})\mathbf{I})$ with $\textstyle{\sigma_{jlk}=\sqrt{{\tau_{p}\mathcal{E}_{p,lk}\beta_{jlk}^{2}}\over{\tau_{p}\sum_{l'}\tau_{p}\mathcal{E}_{p,l'k}\beta_{jl'k}+N_{0}}}.}$
Here, for convenience, we assume that the UEs are numbered such that identically indexed UEs across different cells share the same pilot sequence.

 Following this, the UEs and the BSs simultaneously transmit their data. Let $\hat{\mathbf{v}}_{u,jk}\in\mathbb{C}^{N}$ be the UL combiner $k$th column of the UL combiner vector designed at the $j$th BS. Similarly, let $\hat{\mathbf{v}}_{d,jn}\in\mathbb{C}^{N}$ be the DL precoder designed at the $j$th BS and is intended for the $n$th DL UE. Let the $k$th UL UE of the $j$th cell transmit its symbol $s_{u,jk}$  with power $\mathcal{E}_{u,jk}$, and the $j$th BS transmit  the precoded DL data $\mathbf{v}_{d,jn}s_{u,jn}$. The total power expended by the $j$th BS is denoted by $\mathcal{E}_{d,j}$ and the power control coefficient for the corresponding $n$th UE is denoted by $\kappa_{jn}$. We present the sum UL-DL SE for a cellular FD-mMIMO system with MRC~(i.e. $\mathbf{v}_{u,lk}=\hat{\mathbf{f}}_{u,llk}$) in the UL and MFP~(i.e. $\mathbf{v}_{d,ln}=\hat{\mathbf{f}}_{d,lln}^*$) in the DL in the following Lemma \textcolor{black}{based on~\cite{Asymptotic}:}
 
 \begin{lem}
	\textcolor{black}{The achievable sum UL-DL SE of a cellular FD-mMIMO system with MRC/MFP is} 
	\begin{equation} \textcolor{black}{
	\mathcal{R}^{\text{FD}}_\text{sum}= \frac{\tau-\tau_p}{\tau}\sum\nolimits_{l=1}^{L} \left\{\sum\nolimits_{k\in \mathcal{U}_{u}} \log_{2}(1+\eta_{u,lk}^{\text{FD}})+\sum\nolimits_{n\in\mathcal{U}_{d}}\log_{2}(1+\eta_{d,ln}^{\text{FD}})\right\},}
	\end{equation}
\textcolor{black}{with the UL and DL SINRs being $ \eta_{u,lk}^{\text{FD}}=\frac{N_{r}\sigma_{llk}^{2}\mathcal{E}_{u,lk}}{{\tt IBS}_{jk}+{\tt MUI}_{u,jk}+N_{0}}$, and $\eta_{d,ln}^{\text{FD}}=\frac{N_{t}^2\kappa_{ln}^{ 2}{{\mathcal{E}_{d,l}}}{\sigma_{lln}^4}}{{\tt IUI}_{ln}+{\tt MUI}_{d, ln}+N_{0}}$, respectively, 
with $\textstyle{{\tt IBS}_{jk}\triangleq N_{t}\sum\nolimits_{{j=1, j\neq l}}^{L}\sum\nolimits_{n\in\mathcal{U}_{j,d}}\kappa_{jn}^{ 2}{\mathcal{E}_{d,j}}\rho_{lj}}$,  \\ $\textstyle{{\tt MUI}_{u,jk}\triangleq N_{r}\sum\nolimits_{{j=1, j\neq l}}^{L}\sigma_{ljk}^{2}{\mathcal{E}_{u,jk}}+\sum\nolimits_{j=1}^{L}\sum\nolimits_{k'\in \mathcal{U}_{j,u}}{\beta_{ljk'}\mathcal{E}_{u,jk'}}}$, ${\tt IUI}_{ln}\triangleq \sum\nolimits_{j=1}^{L}\sum\nolimits_{k'\in\mathcal{U}_{j,u}}\!\!\!{\mathcal{E}_{u,lk'}}\epsilon_{jk',ln}$, \\ and ${\tt MUI}_{d,ln} \triangleq N_{t}^2\sum\nolimits_{{j=1, j\neq l}}^{L}\sigma_{jln}^{2}\sigma_{lln}^{2}\mathcal{E}_{d,j}\kappa_{jn}^{2}+ N_{t}\sum\nolimits_{j=1}^{L}\sum\nolimits_{k'\in {\mathcal{U}_{j,d} }}\beta_{jln}\sigma_{jjk'}^{2}{{\mathcal{E}_{d,j}}}\kappa_{jk'}^{ 2}$.}
\end{lem}

\section{Numerical Results}\label{sec:simulation}
In this section, we use our derived results to obtain insights into the performance of DTDD enabled HD-CF mMIMO system. For our numerical experiments, the UEs are dropped uniformly at random locations over a $1$~km $\times \ 1$~km area and are served by $M$ HD-APs, depending on the AP-schedules obtained via Algorithm~\ref{algo:GA_submod}. \textcolor{black}{The APs are arranged in a grid for fair comparison and maximal coverage~\cite{making_cf_emil,EGL2}.}  The pathloss exponent and the reference distance from each AP are assumed to be~$-3.76$ and $10$~m, respectively~\cite{cell_small}.
The UL SNR is set by fixing the noise variance $N_0$ to unity and varying the UL powers $\mathcal{E}_{u,k}$ such that $\mathcal{E}_{u,k}/N_{0}$ equals the desired value. In the DL, we set ${\kappa_{jn}} = (N\sum\limits_{k' \in {{\mathcal{U}}_d}}\alpha _{jk'}^2)^{ - 1}$ in~\eqref{eq: user_DL_receieve}, as in~\cite{Nayebi,cell_small}.  For the cellular system, we partition the area into $L$ equal sized cells with an FD-mMIMO BS deployed at each of the cell centers, and each UE is served by its nearest BS. The numerical results of this section are obtained by considering Monte Carlo simulations over $10^{4}$ random UE location and channel instantiations. 

\vspace*{-0.5cm}
\subsection{Performance comparison with MRC \& MFP:}
\begin{figure*}
	\begin{subfigure}[b]{0.48\textwidth}
	\centering
\includegraphics[width=0.85\textwidth]{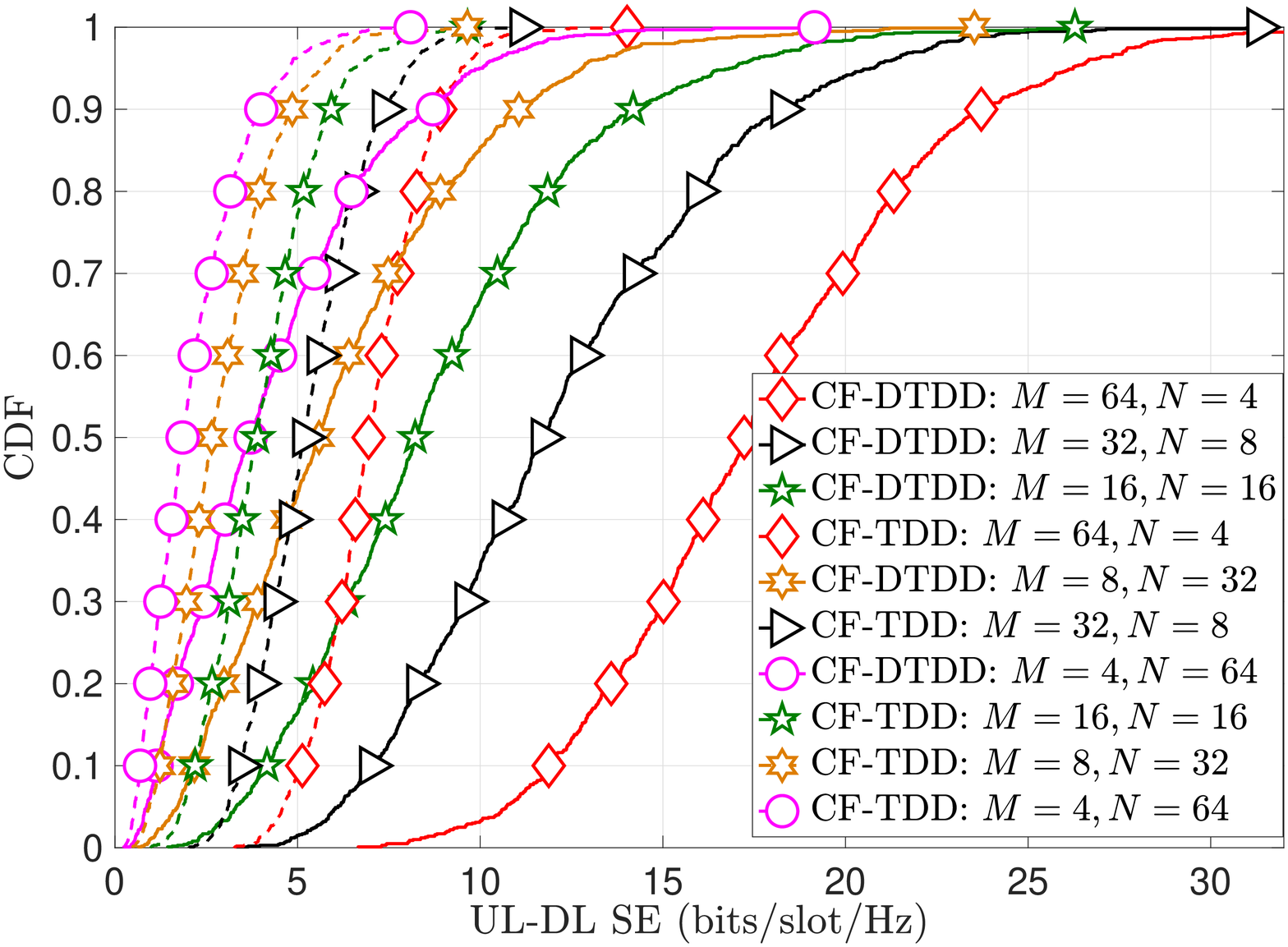}
\caption{CDF of the sum UL-DL SE of DTDD CF-mMIMO and TDD CF-mMIMO with different AP/antenna configurations.}\label{fig:TDD_DTDD_CF_K_32}
	\end{subfigure}
	\begin{subfigure}[b]{0.48\textwidth}
	\centering
\includegraphics[width=0.85\textwidth]{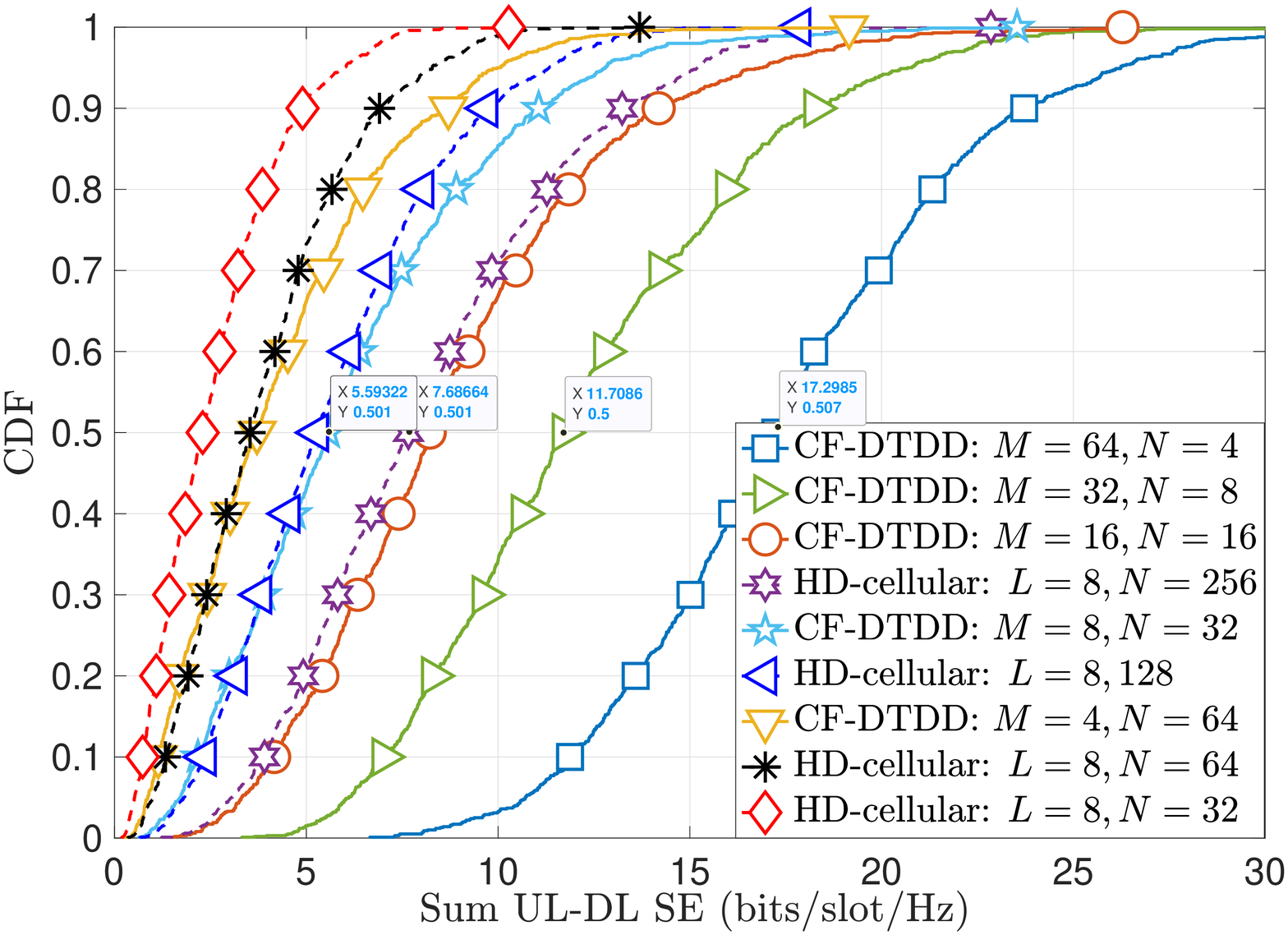}
\caption{CDF of the sum UL-DL SE of DTDD CF-mMIMO and TDD enabled cellular mMIMO.}\label{fig:cTDD_DTDD_CF_K_32}
\end{subfigure}
\caption{\textcolor{black}{Comparison of DTDD  CF-mMIMO with canonical TDD based CF and cellular mMIMO, with MRC/MFP employed at each of the APs/BSs.}
	}\label{fig:DTDD_TDD}
\vspace*{-.5cm}
\end{figure*}

\textcolor{black}{In Fig.~\ref{fig:DTDD_TDD}, we compare the performance of DTDD  CF-mMIMO with TDD based HD cellular and CF mMIMO, via the CDFs of the sum UL-DL SE. We consider $K=32$ UEs with $50\%$ of the UEs having UL data demand in each time slot. For each instantiation of UE positions, the APs are scheduled using the proposed greedy algorithm. In the cellular case, we consider $L=8$, with the BS in each cell equipped with $N$ antennas. For the HD TDD-CF and DTDD enabled CF, we consider multiple combinations $M$ and $N$. From  Fig.~\ref{fig:TDD_DTDD_CF_K_32}, we see that DTDD enabled CF-mMIMO system considerably improves the sum UL-DL SE compared to the other schemes. For example, CF-DTDD with $(M=64, N=4)$ offers a median sum UL-DL SE of $17$~bits/slot/Hz, whereas TDD CF offers only~$7$~bits/slot/Hz.   Next, in Fig.~\ref{fig:cTDD_DTDD_CF_K_32}, we compare DTDD CF-mMIMO with cellular TDD mMIMO. Cellular TDD with $(L=8, N= 256)$ performs similar to DTDD CF-mMIMO with $(M=16, N=16)$; note that the antenna density in the cellular system is $8$ times the antenna density of the CF system. DTDD schedules the APs based on the localized traffic demand and the UL-DL transmissions occur simultaneously, which results in the dramatic improvement in the system sum UL-DL SE compared to the cellular TDD case.
}

\begin{figure*}
	\begin{subfigure}[b]{0.48\textwidth}
		\centering
\includegraphics[width=0.85\textwidth]{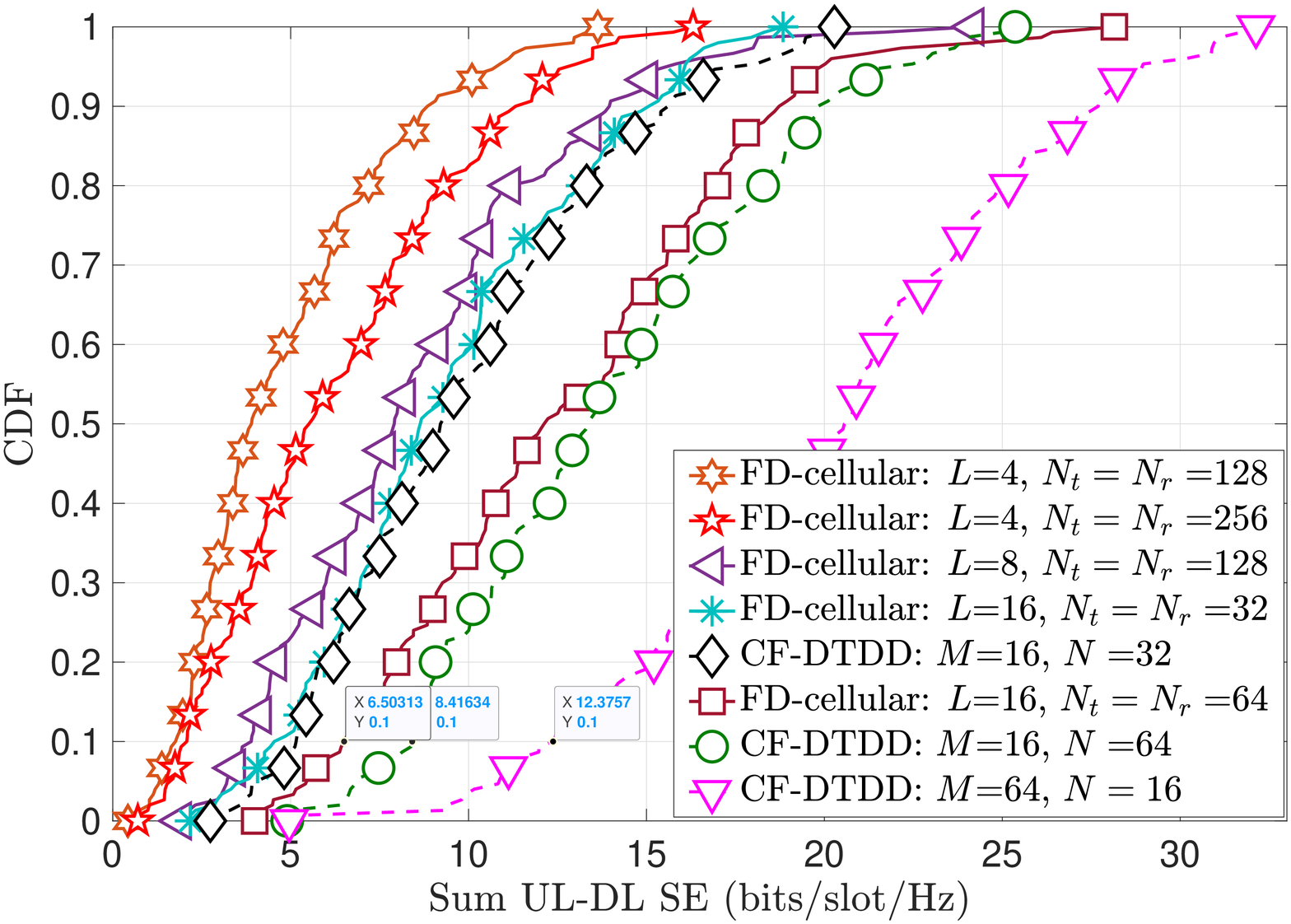}\caption{\textcolor{black}{CDF of the sum UL-DL SE of a DTDD CF-mMIMO and a cellular FD-mMIMO system with $K=32$ UEs.}}\label{fig:CDF_sum_SE}
	\end{subfigure}\hfill
	\begin{subfigure}[b]{0.48\textwidth}
		\centering
\includegraphics[width=0.85\textwidth]{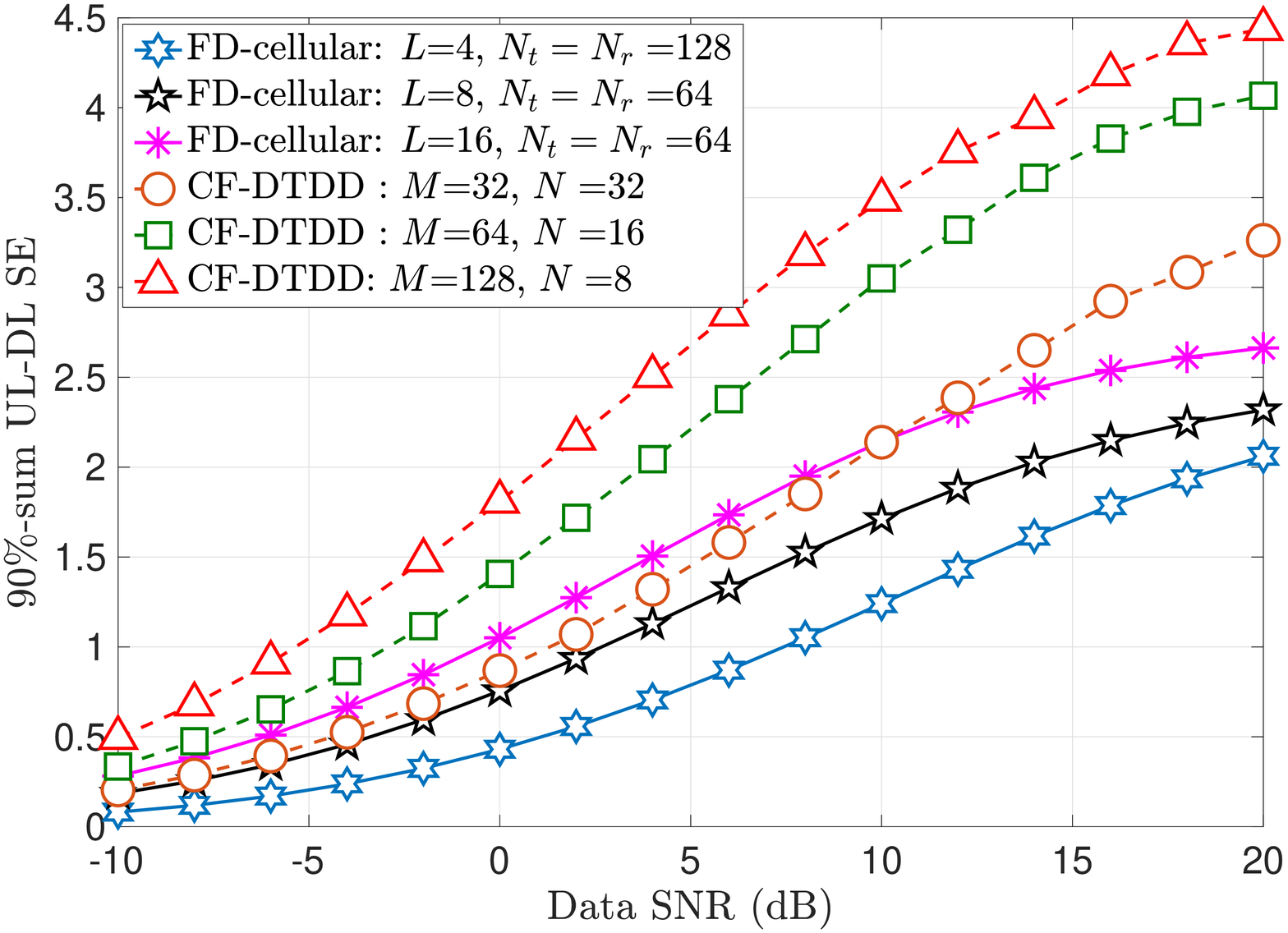}
\caption{$90\%$-likely sum UL-DL SE vs. UL and DL data SNR, with $K=60$ UEs. }\label{fig:antenna_denity}
\end{subfigure}\caption{\textcolor{black}{Comparison of DTDD CF-mMIMO  with an FD-cellular system.}}
\vspace*{-0.9cm}
\end{figure*}

\textcolor{black}{Next, in Fig.~\ref{fig:CDF_sum_SE}, we compare DTDD CF-mMIMO with an FD cellular system. CF-DTDD with HD APs and $(M=16,N=64)$ outperforms the cellular FD-system with double the antenna density, i.e., $(L=16, N_t=N_r=64)$. Increasing the number of APs, but still with half the antenna density compared to the FD (see the curve corresponding to $(M=64, N=16)$), results in significantly better sum UL-DL SE in HD CF-DTDD compared to the cellular FD system. 
Thus, although each BS in cellular system is equipped with simultaneous transmit and receive capability, the HD-distributed APs with dynamic scheduling and the joint processing benefits of a DTDD CF-mMIMO results in better sum UL-DL SE in the system.} 

\textcolor{black}{Next, we illustrate the dependence of the sum UL-DL SE on the data SNR.} In Fig.~\ref{fig:antenna_denity}, we plot the average $90\%$-likely sum UL-DL SE as a function of the UL and DL data SNR. We observe that at low data SNR regime~($-10$ to $10$ dB) an FD-cellular system with $(L=16, N_t=N_r=64)$ offers similar $90\%$-likely UL-DL SE compared to the CF-DTDD system with half the antenna density $(M=32, N=32)$. Moreover, if we increase the number of APs deployed, for example $(M=64, N=16), (M=128, N=8)$, CF-DTDD offers better performance throughout the entire range of data SNR. In both the cases, for a given antenna density, having a larger number of BS/APs is better: the beamforming gains are insufficient to offset the path loss and interference. 

\textcolor{black}{In Fig.~\ref{fig:data_load1}, we show the trade-off between the pilot length and the available data duration via plotting the $90\%$-likely sum UL-DL SE as a function of the ratio of the number of active UEs  to the coherence interval.} We consider two cases: $(i)$ $\tau_{p}=30$ irrespective of the number of UEs in the system~(Fig.~\ref{fig:data_load11}), $(ii)$ $\tau_{p}=K$, i.e., the pilot length is scaled linearly with the UE load~(Fig.~\ref{fig:data_load12}). We consider the overall fractional UL-DL data demands to be the same across the number of UEs. In case $(i)$, the sum UL-DL SE increases monotonically, even though there is pilot contamination in the system. This shows the effectiveness of the iterative pilot allocation algorithm presented in Sec.~\ref{sec: pilot_allocation}. However, in case $(ii)$, the duration available for data transmission reduces, leading to a decrease in the SE as the number of UEs increases. For instance, in Fig.~\ref{fig:data_load12}, with $(M=64, N=2)$,  the sum UL-DL SE decreases sharply when UE load goes beyond $55\%$ of the coherence interval. Thus, as the UE load increases, it is better to repeat shorter length pilots, along with a suitable algorithm to ensure minimal pilot contamination, to balance the errors introduced by pilot contamination with the data transmission duration.

\begin{figure}
	\begin{subfigure}[b]{0.49\textwidth}
		\centering
		\includegraphics[width=0.87\textwidth]{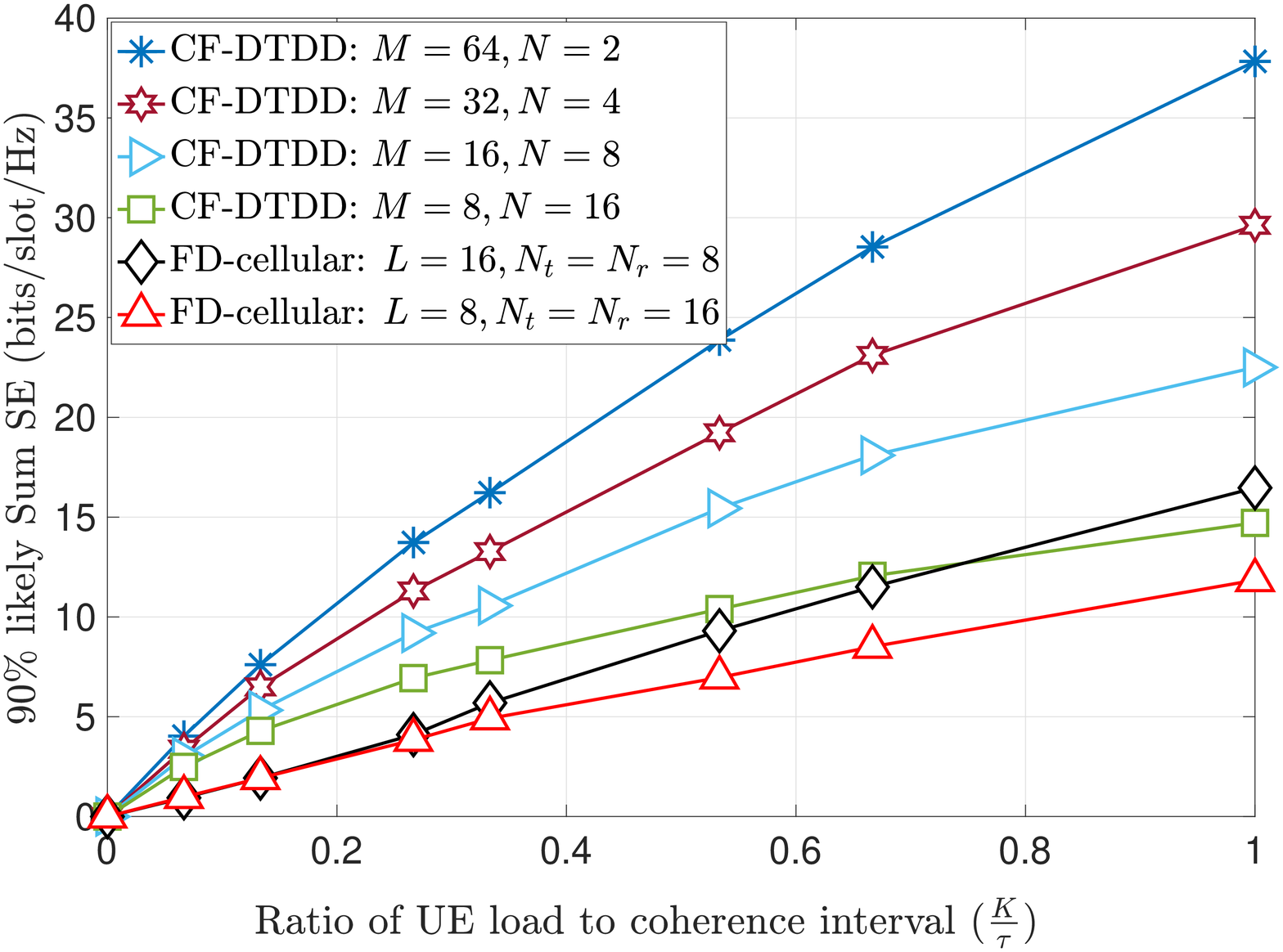}
		\caption{The $90\%$-likely sum UL-DL SE vs. the number of active UEs. with $\tau_p=30$,  $\tau=200$.}\label{fig:data_load11}
		\end{subfigure}
	\hfill
		\begin{subfigure}[b]{0.49\textwidth}
			\centering
		\includegraphics[width=0.87\textwidth]{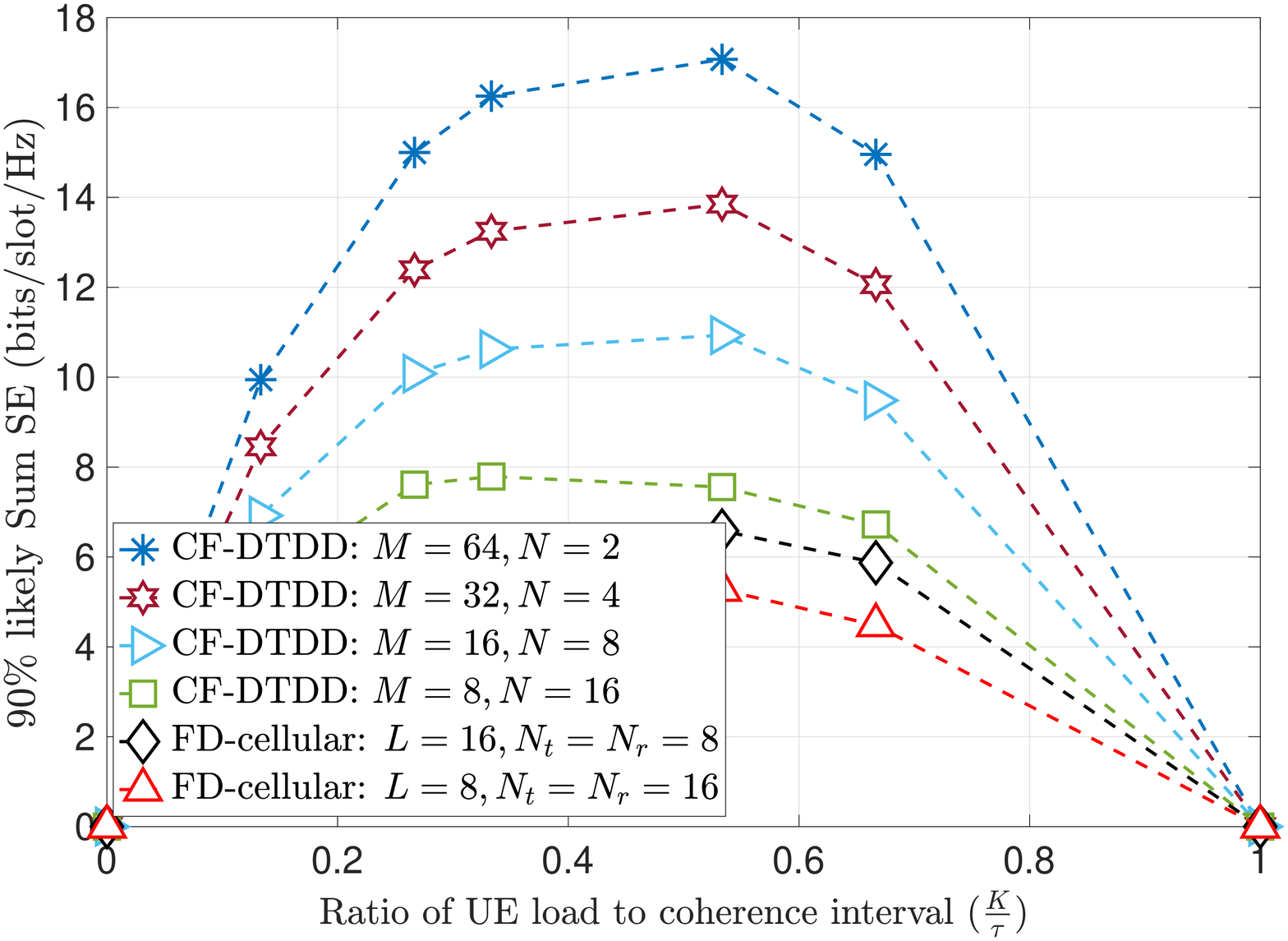}
		\caption{The $90\%$-likely sum UL-DL SE vs. the number of active UEs. with $\tau_p=K$, $\tau=200$.}\label{fig:data_load12}
	\end{subfigure}
\caption{The $90\%$-likely sum UL-DL SE vs. the number of active UEs. We see that scaling the pilot length with the number of UEs can be detrimental at high UE load; necessitating careful pilot allocation with limited length.}\label{fig:data_load1}
\vspace*{-0.7cm}
\end{figure}

Next in Fig.~\ref{fig:RSI}, we investigate the effect of AP-AP and inter-BS CLI on the UL sum SE in CF and cellular mMIMO systems, respectively. Here, we fix the DL SNR to $10$~dB for all UEs. We observe that in a cellular FD system, the UL sum SE reduces dramatically when the inter-BS CLI exceeds $-40$~dB. In contrast, in the DTDD enabled HD CF-mMIMO system, as the AP-AP CLI increases, the greedy algorithm ensures an AP-schedule that balances the UL and DL SE to maximize the overall sum SE. For instance, with $(M=32,N=4)$, as CLI increases from $-70$ dB to $-60$ dB, we observe a decrease in the UL SE, and beyond $-20$ dB, it saturates to about $5$ bits/slot/Hz. In contrast, in an FD cellular system with $(L=8, N_{t}=N_{r}=16)$ the UL SE reduces to nearly $0$ bits/slot/Hz at $-20$ dB of inter-BS CLI. Therefore, HD-APs with DTDD are more resilient to SI cancellation errors.  Also, the performance of the cellular FD-mMIMO shown in Fig.~\ref{fig:RSI} is an upper bound, since we consider perfect SI cancelation at the BSs.

\begin{figure*}
		\begin{subfigure}[b]{0.49\textwidth}
		\centering
		\includegraphics[width=0.87\textwidth]{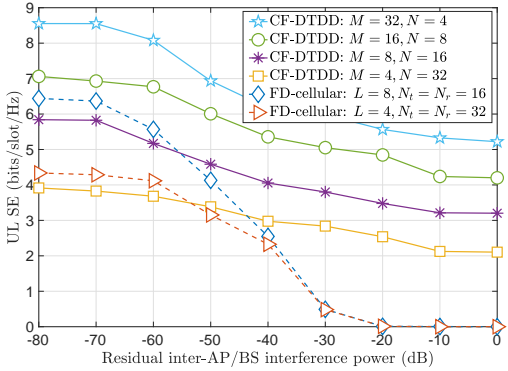}
		\caption{The average UL sum SE vs residual BS-BS/AP-AP cross link interference power. CF-mMIMO system with HD APs and DTDD is more resilient to CLI.}
		\label{fig:RSI}
	\end{subfigure}\hfill
	\begin{subfigure}[b]{0.49\textwidth}
\centering
\includegraphics[width=0.87\textwidth]{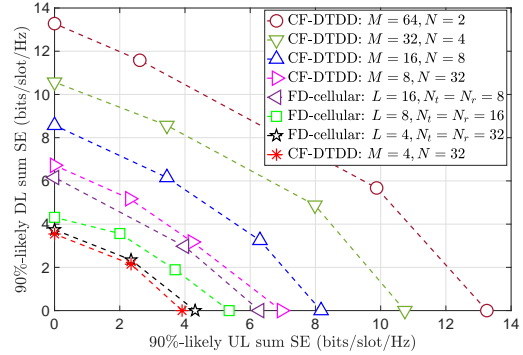}
\caption{Rate region between $90\%$-likely UL sum SE vs. $90\%$-likely DL sum SE. We observe that more APs with smaller antennas provides larger rate region compared to cellular FD system.}
\label{fig:rate_region}
	\end{subfigure}
\caption{Effect of SI on the UL SE and illustration of rate region under DTDD based CF-mMIMO system.}
\vspace*{-0.47cm}
\end{figure*}

In Fig.~\ref{fig:rate_region}, to illustrate the effect of the traffic demand on the UL and DL SE, we vary the fraction of UEs demanding UL data from $0$ to $1$, and plot the  $90\%$-likely UL sum SE against the $90\%$-likely DL sum SE obtained for each fractional UL-data demand. At each fractional UL-data demand, we use Algorithm~\ref{algo:GA_submod} to determine the mode of operation across the APs. We observe that the rate region attained by HD CF-DTDD is significantly larger than that of the cellular FD system, e.g., the HD CF-DTDD curve with $(M=16, N=8)$ and the FD-cellular curve with $(L=16, N_t=N_r=8)$. In fact, even with $100\%$ UL data demand (the points along the x-axis) or $100\%$ DL data demand (the points on the y-axis),  HD CF-DTDD outperforms FD-cellular by more than $2$~bits/slot/Hz. This is because the FD-cellular system has to contend with inter-cell interference, even if the SI cancelation is perfect. The joint data processing at the CPU and dynamic AP-scheduling based on the local data demands results in the higher  rate region obtained by the HD CF-DTDD mMIMO system.

\vspace*{-0.3cm}
\subsection{Performance comparison with MMSE \& RZF:}

\textcolor{black}{In Fig.~\ref{fig:perfoamnce_comp_MMSE_RZF_TDD_DTDD}, we compare the TDD based canonical CF-mMIMO system with our proposed DTDD enabled CF-mMIMO system. For both the schemes, we consider a centralized MMSE combiner in the UL and RZF in the DL.  DTDD CF-mMIMO with $(M=64, N=2)$ procures a sum UL-DL SE of $14.63$ bits/slot/Hz, while in similar settings, the TDD based CF-system only obtains $9.83$ bits/slot/Hz. Thus, although MMSE-type combiners and precoders improve the SE of both TDD and DTDD based system, the simultaneous UL-DL data traffic handling capabilities of DTDD based system help further enhance the achievable sum UL-DL SE.}
\begin{figure*}
\begin{subfigure}[b]{0.47\textwidth}
	\centering
	\includegraphics[width=0.87\textwidth]{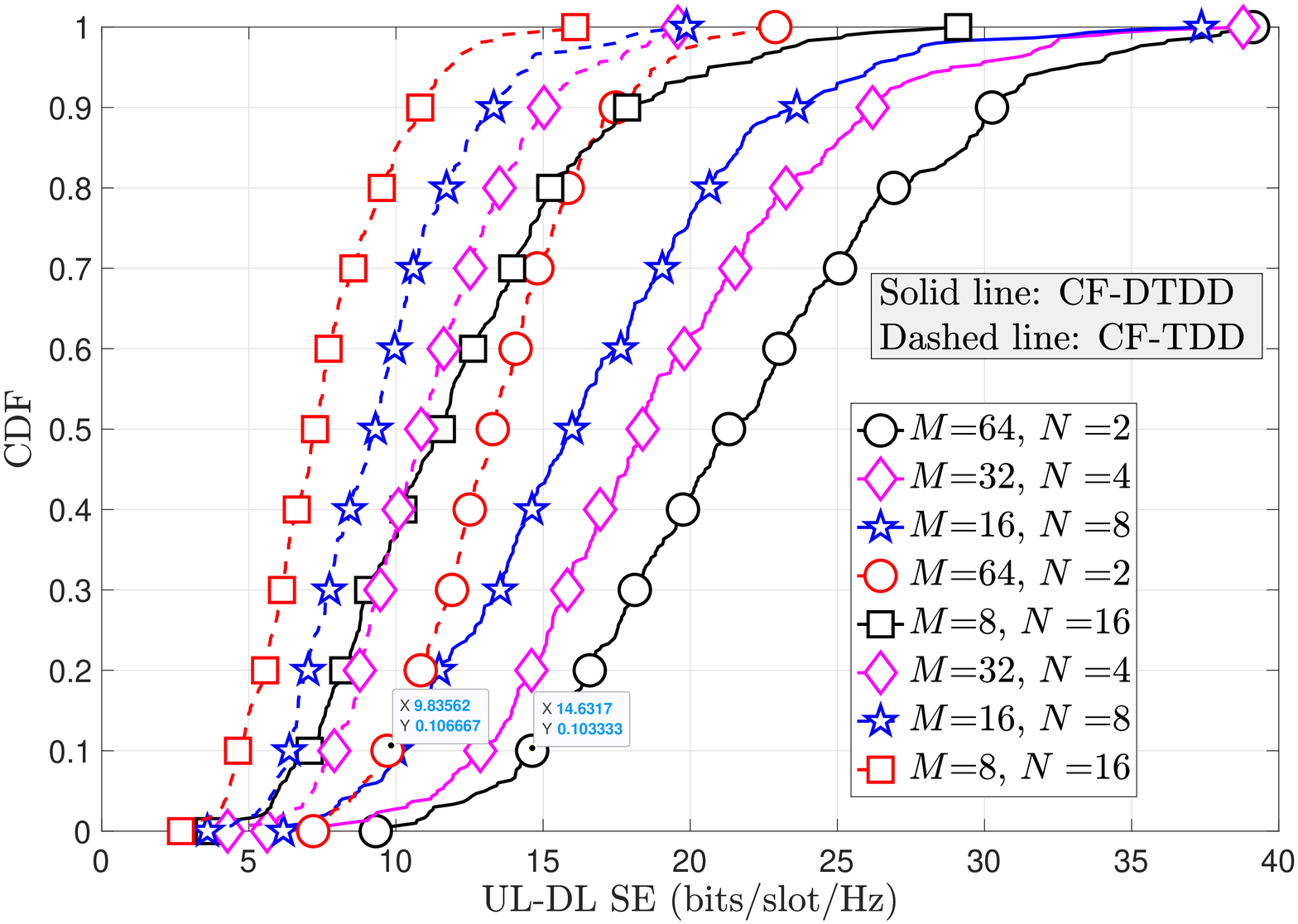}
	\caption{\textcolor{black}{Comparison of DTDD CF-mMIMO  and TDD based CF-mMIMO.}}\label{fig:perfoamnce_comp_MMSE_RZF_TDD_DTDD}
\end{subfigure}
\hfill
\begin{subfigure}[b]{0.47\textwidth}
	\centering
	\includegraphics[width=0.87\textwidth]{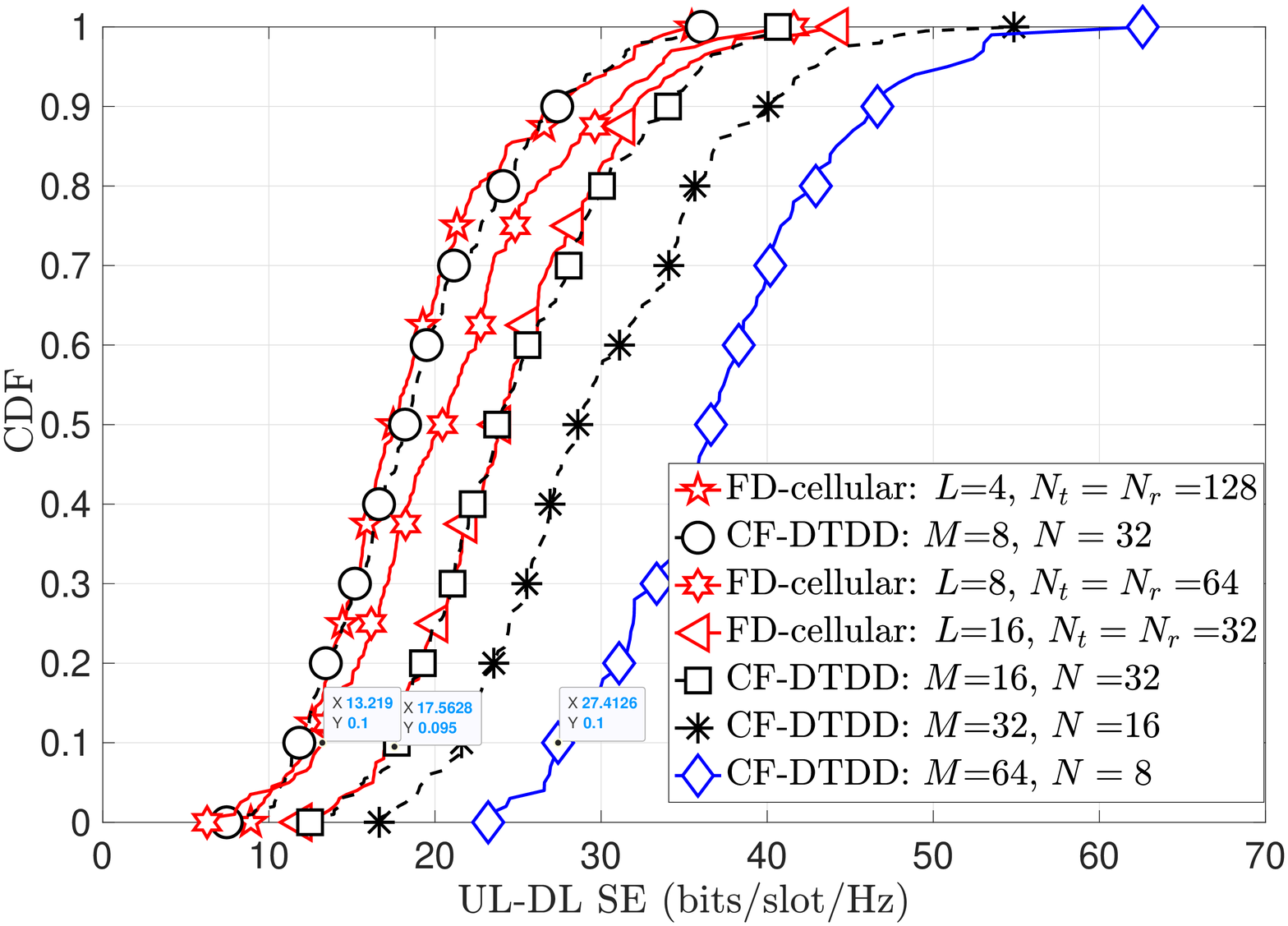}
	\caption{\textcolor{black}{CDF of the sum UL-DL SE achieved via DTDD CF-MIMO and FD mMIMO systems.}}\label{fig:MMSE_RZF_FD_CF}
\end{subfigure}\caption{\textcolor{black}{Comparison of DTDD CF-mMIMO  with CF-TDD and an FD cellular mMIMO systems.}}
\vspace*{-0.87cm}
\end{figure*}

\textcolor{black}{Next, we compare the performance of an FD-enabled cellular mMIMO system with the CF-DTDD based system. For the cellular case, we consider the multi-cell MMSE~(M-MMSE) combiner and precoder~\cite{massivemimobook}. From Fig~\ref{fig:MMSE_RZF_FD_CF},  we observe that an FD cellular system with $(L=16, N_{t}=N_{r}=32)$ and a CF-DTDD system with $(M=16, N=32)$ have a similar CDF of the sum UL-DL SE, in spite of the CF-system having half the antenna density as the FD system. Also, in the CF system, since the APs are HD, only a subset of the $16$ APs serve the UL UEs and its complement serves the DL APs. In contrast, all the $16$ FD BS can simultaneously serve both UL and DL UEs in their respective cells.  
We further see that an FD-system with $4$ BSs having $128$ transmit and receive antennas each offers a $90\%$ sum UL-DL SE of $13.2$ bits/slot/Hz, whereas the CF-DTDD based system with $(M=64, N=8)$ offers  $27$ bits/slot/Hz, a more than $100\%$ improvement. Although both FD cellular and CF-DTDD system support simultaneous UL and DL traffic load, the joint signal processing capability at the CPU in the CF architecture helps to substantially improve the sum UL-DL SE.}

\section{Conclusions}

In this paper, we analyzed the performance of an CF-mMIMO system under DTDD where the transmission mode at each HD AP is scheduled so that the sum UL-DL SE is maximized. The complexity of the brute-force search increases exponentially with the number of APs. To tackle this problem, we developed a sub-modularity based  greedy algorithm with associated optimality guarantees. Our numerical experiments revealed that a DTDD enabled CF-mMIMO system substantially improves the sum UL-DL SE compared to HD TDD CF and HD TDD cellular mMIMO system. \textcolor{black}{The key reason of the performance improvement in  DTDD compared to TDD based systems is that the former duplexing scheme can simultaneously serve the UL and the DL UEs in the system.} Furthermore, the HD DTDD CF-mMIMO can even outperform an FD-cellular mMIMO system. \textcolor{black}{Through extensive experiments, we illustrated the sum SE improvement of DTDD CF-mMIMO over the cellular FD system under different antenna densities, number of UEs, and fractional UL/DL data demands. We considered the MRC/MFP based scheme as well as the centralized MMSE/RZF based scheme, and in both cellular and CF systems. Under all these different system settings, we showed that CF-DTDD with a large number of APs can improve the sum UL-DL SE compared to an FD cellular system under similar antenna densities.} Essentially,  DTDD CF-mMIMO exploits the joint signal processing of a CF system coupled with the adaptive scheduling of UL-DL slots based on the localized traffic demands at the APs. We also presented an iterative pilot allocation algorithm which substantially reduces the effect of pilot contamination. A key advantage of DTDD enabled CF over the FD cellular system is that we no longer need additional hardware at each AP to cancel the SI. 
The system performance can be further improved by incorporating UL-to-DL UE interference cancellation techniques or power control strategies. Another aspect that is worth investigating further is a theoretical analysis of the latency performance of DTDD enabled CF-mMIMO systems, \textcolor{black}{and weighted sum SE maximization with UL/DL fairness constraints. As a final note, an FD-enabled CF-mMIMO system can potentially outperform the HD DTDD CF system, at the cost of SI cancelation hardware at the APs. It would be interesting to theoretically analyze their performance.} 
\vspace*{-0.3cm}
\appendices 
\begin{appendices}
\renewcommand{\thesectiondis}[2]{\Alph{section}:}
\vspace*{-0.27cm}
\section{Proof of Theorem~\ref{thm: rate_UL_thm}}\label{appen: rate_uplink_thm}
\begin{proof}
	The numerator of~\eqref{eq: MRC_uplink} can be written as $\mathbb{E}[\sum\nolimits_{m\in\mathcal{A}_{u}} \hat{\mathbf{f}}_{mk}^{H}{\hat{\mathbf{f}}_{mk}+\hat{\mathbf{f}}_{mk}^{H}\tilde{\mathbf{f}}_{mk}} ]=\sum\limits_{m\in\mathcal{A}_{u}}\mathbb{E}[\|\hat{\mathbf{f}}_{mk}\|^{2}]=N\sum\nolimits_{m\in\mathcal{A}_{u}}\alpha_{mk}^{2}.$ 
	Next, we can show ${\tt var}[\sum\nolimits_{m\in\mathcal{A}_{u}}\hat{\mathbf{f}}_{mk}^{H}{\mathbf{f}}_{mk}]=\sum\nolimits_{m\in\mathcal{A}_{u}}N\alpha_{mk}^2\beta_{mk}$.
	
	Now, for the UEs that share their pilot sequences  with the $k$th UE, i.e., for $n\in\mathcal{I}_{p}\backslash k$,
	\vspace*{-.4cm}
	\begin{align}\label{eq:MUI_interm_1}
	&\textstyle{\mathbb{E}\big[\big|\sum\limits_{m\in\mathcal{A}_{u}}\hat{\mathbf{f}}_{mk}^{H}\mathbf{f}_{mn}\big|^2\big]=\mathbb{E}\big[\big|\sum\limits_{m\in\mathcal{A}_{u}}c_{mk}\sqrt{\tau_{p}\mathcal{E}_{p,k}}\beta_{mk}\big(\hspace{-.1cm}\sum\limits_{{n'\in \mathcal{I}_{p}}}\sqrt{\mathcal{E}_{p,n'}\tau_{p}}\mathbf{f}_{mn'}+\dot{\mathbf{w}}_{p,m}\big)^{H}\mathbf{f}_{mn}\big|^2\big]}\notag\\
	&\textstyle{=N\sum\limits_{m\in\mathcal{A}_{u}}N_{0}c_{mk}^2{\tau_{p}\mathcal{E}_{p,k}}\beta_{mk}^2\beta_{mn}+\mathbb{E}\big[\big|\sum\limits_{m\in\mathcal{A}_{u}}c_{mk}\sqrt{\tau_{p}\mathcal{E}_{p,k}}\beta_{mk}\sum\limits_{n'\in \mathcal{{I}}_{p}}\sqrt{\mathcal{E}_{p,n'}\tau_{p}}\mathbf{f}_{mn'}^{H}\mathbf{f}_{mn}\big|^2\big]}.
	\end{align}
	The last term in the above can be simplified as follows:
	\vspace*{-.4cm}
	\begin{align*}
	&\textstyle{\mathbb{E}\big[\big|\sum\nolimits_{m\in\mathcal{A}_{u}}c_{mk}\sqrt{\tau_{p}\mathcal{E}_{p,k}}\beta_{mk}\big(\sum\nolimits_{n'\in \mathcal{{I}}_{p}\setminus n}\sqrt{\mathcal{E}_{p,n'}\tau_{p}}\mathbf{f}_{mn'}+\sqrt{\mathcal{E}_{p,n}\tau_{p}}\mathbf{f}_{mn}\big)^{H}\mathbf{f}_{mn}\big|^2\big]}
	\notag\\&=\textstyle{N \!\!\! \sum\limits_{m\in\mathcal{A}_{u}}\!\!\! c_{mk}^2{\tau_{p}^2\mathcal{E}_{p,k}}\beta_{mk}^{2}\mathcal{E}_{p,n}\beta_{mn}^{2}\!\!+\!\! N^{2}(\!\sum\limits_{m\in\mathcal{A}_{u}}\!\!\!\alpha_{mk}^2\frac{\beta_{mn}}{\beta_{mk}}\sqrt{\frac{\mathcal{E}_{p,n}}{\mathcal{E}_{p,k}}})^{2}}\!\!+\!\!\textstyle{N\!\!\!\sum\limits_{m\in\mathcal{A}_{u}}\!\sum\limits_{n'\in \mathcal{{I}}_{p}\setminus n}\!\!\!\!\! c_{mk}^2{\tau_{p}^2\mathcal{E}_{p,k}}\beta_{mk}^{2}\mathcal{E}_{p,n'}\beta_{mn'}\beta_{mn}}.
	\end{align*}
	Now $\textstyle{N\sum\nolimits_{m\in\mathcal{A}_{u}}\sum\nolimits_{n'\in \mathcal{{I}}_{p}\setminus n}c_{mk}^2{\tau_{p}^2\mathcal{E}_{p,k}}\beta_{mk}^{2}\mathcal{E}_{p,n'}\beta_{mn'}\beta_{mn}}$,
	%~\eqref{eq:MUI_interm_2} 
	can be  further expanded as \vspace{-0.2cm}
	\begin{align}\label{eq:MUI_interm_3}
	&
	\textstyle{N\sum\limits_{m\in\mathcal{A}_{u}}\hspace*{-.12cm}c_{mk}\sqrt{\tau_{p}\mathcal{E}_{p,k}}\beta_{mk}\beta_{mn}\big\{c_{mk}\sqrt{\tau_{p}\mathcal{E}_{p,k}}\beta_{mk}\sum\limits_{n'\in \mathcal{{I}}_{p}}\tau_{p}\mathcal{E}_{p,n'}\beta_{mn'}\big\}\hspace*{-.1cm}-N\hspace*{-.12cm}\sum\limits_{m\in\mathcal{A}_{u}}\hspace*{-.14cm}c_{mk}^2{\tau_{p}^2\mathcal{E}_{p,k}}\beta_{mk}^2\mathcal{E}_{p,n}\beta_{mn}^2}\notag\\&\textstyle{= N\hspace*{-.12cm}\sum\limits_{m\in\mathcal{A}_{u}}\hspace*{-.12cm}{\tau_p\mathcal{E}_{p,k}}c_{mk}\beta_{mk}^2\beta_{mn}-NN_{0}\hspace*{-.12cm}\sum\limits_{m\in\mathcal{A}_{u}}\hspace*{-.12cm}c_{mk}^2{\tau_{p}\mathcal{E}_{p,k}}\beta_{mk}^{2}\beta_{mn}-N\hspace*{-.12cm}\sum\limits_{m\in\mathcal{A}_{u}}\hspace*{-.12cm}c_{mk}^2{\tau_{p}^2\mathcal{E}_{p,k}}\beta_{mk}^2\mathcal{E}_{p,n}\beta_{mn}^2},
	\end{align}
	where, in the last step above, we used the fact $ c_{mk}\sqrt{\tau_{p}\mathcal{E}_{p,k}}\beta_{mk}\sum_{n'\in \mathcal{{I}}_{p}}\tau_{p}\mathcal{E}_{p,n'}\beta_{mn'}=\sqrt{\tau_{p}\mathcal{E}_{p,k}}\beta_{mk}-N_{0}c_{mk}\sqrt{\tau_{p}\mathcal{E}_{p,k}}\beta_{mk}$. Combining~\eqref{eq:MUI_interm_1} and~\eqref{eq:MUI_interm_3}, with $\alpha_{mk}^2={\tau_p\mathcal{E}_{p,k}}c_{mk}\beta_{mk}^2$, we obtain
	\vspace*{-.4cm}
	\begin{equation}\label{eq:Coh_1}
	\textstyle{\sum\limits_{\substack{n\in \mathcal{I}_{p}\backslash k}}{\mathcal{E}_{u,n}}\mathbb{E}\big[\big|\sum\limits_{m\in\mathcal{A}_{u}}\hat{\mathbf{f}}_{mk}^{H}\mathbf{f}_{mn}\big|^2\big]=\!\!\! \sum\limits_{\substack{n\in \mathcal{I}_{p}\backslash k}}\!\!\!{\mathcal{E}_{u,n}}\big(N^{2}\big(\sum\limits_{m\in\mathcal{A}_{u}}\alpha_{mk}^2\frac{\beta_{mn}}{\beta_{mk}}\sqrt{\frac{\mathcal{E}_{p,n}}{\mathcal{E}_{p,k}}}\big)^{2}\!\!\!+\!\! N\!\!\!\sum\limits_{m\in\mathcal{A}_{u}}\alpha_{mk}^2\beta_{mn}\big)}.
	\end{equation}
The first two terms above correspond to coherent interference~\eqref{eq:Coh} and  non-coherent interference, respectively, from UEs that share the $k$th UE's pilot. Next, considering the interference due to the UEs that do not share the $k$th UE's pilot, we obtain
\vspace*{-.4cm}
	\begin{equation}\label{eq:NCoh_1}
		\textstyle{\sum\nolimits_{q\in \mathcal{U}_{u}\backslash\mathcal{I}_{p} }{\mathcal{E}_{u,n}}\mathbb{E}\big[\big|\sum\nolimits_{m\in\mathcal{A}_{u}}\hat{\mathbf{f}}_{mk}^{H}\mathbf{f}_{mq}\big|^{2}\big]=N\sum\nolimits_{q\in \mathcal{U}_{u}\backslash\mathcal{I}_{p} }{\mathcal{E}_{u,n}}\sum\nolimits_{m\in\mathcal{A}_{u}}\tau_p\mathcal{E}_{p,k}c_{mk}\beta_{mk}^2\beta_{mq}}.
	\end{equation}
	Hence,~\eqref{eq:NCoh} follows via combining the beamforming uncertainty, the second term of~\eqref{eq:Coh_1}, and~\eqref{eq:NCoh_1}.
	Finally, we can derive the inter-AP interference as 
	\vspace*{-.4cm}
	\begin{align*}
		&\textstyle{\sum\limits_{n\in\mathcal{U}_{d}}\mathbb{E}\big[\big|\sum\limits_{m\in\mathcal{A}_{u}}\sum\limits_{j\in\mathcal{A}_{d}}\sqrt{\mathcal{E}_{d,j}}\kappa_{jn} \hat{\mathbf{f}}_{mk}^{H}\mathbf{G}_{mj}\hat{\mathbf{f}}_{jn}^{*}\big|^{2}\big]}
		\textstyle{=\sum\limits_{n\in\mathcal{U}_{d}}\sum\limits_{m\in\mathcal{A}_{u}}\sum_{j\in \mathcal{A}_{d}}{\mathcal{E}_{d,j}}\kappa_{jn}^2\mathbb{E}\big[\text{tr}\big(\mathbf{G}_{mj}\hat{\mathbf{f}}_{jn}^{*}\hat{\mathbf{f}}_{jn}^{T}\mathbf{G}_{mj}^{H}\hat{\mathbf{f}}_{mk}\hat{\mathbf{f}}_{mk}^{H}\big)\big]}\notag\\&\textstyle{{(a)\atop=}\sum\limits_{n\in\mathcal{U}_{d}}\sum\limits_{m\in\mathcal{A}_{u}}\sum\limits_{j\in \mathcal{A}_{d}}{\mathcal{E}_{d,j}}\kappa_{jn}^2\text{tr}\big(\mathbb{E}\big[\mathbf{G}_{mj}\hat{\mathbf{f}}_{jn}^{*}\hat{\mathbf{f}}_{jn}^{T}\mathbf{G}_{mj}^{H}\big]\mathbb{E}\big[\hat{\mathbf{f}}_{mk}\hat{\mathbf{f}}_{mk}^{H}\big]\big)\!\!=\!\!\! N^2\!\!\!\!\sum\limits_{m\in\mathcal{A}_{u}}\sum\limits_{j\in\mathcal{A}_{d}}\sum\limits_{n\in\mathcal{U}_{d}}\kappa_{jn}^2\zeta_{mj}\alpha_{mk}^2\alpha_{jn}^2\mathcal{E}_{d,j}}.
	\end{align*}  
	In $(a)$, we apply the linearity of trace. Then we use the fact that $\mathbf{G}_{mj}$  and $\mathbf{f}_{jn}$ are independent, and therefore, $\mathbb{E}[\mathbf{G}_{mj}\hat{\mathbf{f}}_{jn}^{*}\hat{\mathbf{f}}_{jn}^{T}\mathbf{G}_{mj}^{H}]=\zeta_{mj}\mathbb{E}[\text{tr}(\hat{\mathbf{f}}_{jn}^{*}\hat{\mathbf{f}}_{jn}^{T})]\mathbf{I}_{N}=N\alpha_{jn}^2\zeta_{mj}\mathbf{I}_{N}$, and $\mathbb{E}[\hat{\mathbf{f}}_{mk}\hat{\mathbf{f}}_{mk}^{H}]=\alpha_{mk}^2\mathbf{I}_{N}$, which yields the final result in~\eqref{eq:IIAP}.
\end{proof}
	\vspace*{-.47cm}
\section{Proof of Theorem~\ref{thm:rate_final_dllink}}\label{appen: rate_dllink_thm}
\begin{proof}
	We note that from~\eqref{eq: user_DL_receieve}, we can write the effective DL SINR as
		\begin{align}\label{eq:R_dl_ergodic}
			&\textstyle{\eta_{d,n}\hspace*{-.1cm}=\hspace*{-.1cm}\big[{\big|\sum\nolimits_{j\in\mathcal{A}_{d}}\kappa_{jn}\sqrt{\mathcal{E}_{d,j}}\mathbb{E}\big[\mathbf{f}_{jn}^{T}\hat{\mathbf{f}}_{jn}^{*}\big]\big|^2}\big]\hspace*{-.1cm}\times\hspace*{-.1cm}\big(\hspace*{-.1cm}{\tt var}\big\lbrace\hspace*{-.1cm}\sum\nolimits_{~j\in\mathcal{A}_{d}}\hspace*{-.15cm}\kappa_{jn}\sqrt{\mathcal{E}_{d,j}}\mathbf{f}_{jn}^{T}\hat{\mathbf{f}}_{jn}^{*}\hspace*{-.1cm}\big\rbrace\hspace*{-.1cm}}\notag\\&+\textstyle{\hspace{-.3cm}\sum\limits_{\substack{q\in \mathcal{I}_{p}\backslash n}}\hspace{-.15cm}\mathbb{E}\big[\big|\sum\limits_{j\in\mathcal{A}_{d}}\hspace*{-.1cm}\kappa_{jq}\sqrt{\mathcal{E}_{d,j}}\mathbf{f}_{jn}^{T}\hat{\mathbf{f}}_{jq}^{*}\big|^{2}\big]\hspace{-0.1cm}+\hspace{-0.3cm}\sum\limits_{q\in\mathcal{U}_d\backslash \mathcal{I}_{p}}\!\!\!\mathbb{E}\big[\big|\sum\limits_{j\in\mathcal{A}_{d}}\kappa_{jq}\sqrt{\mathcal{E}_{d,q}}\mathbf{f}_{jn}^{T}\hat{\mathbf{f}}_{jq}^{*}\big|^{2}\big]\hspace{-0.2cm}+\hspace{-0.2cm}\sum\limits_{k\in\mathcal{U}_{u}}\mathcal{E}_{u,n}\mathbb{E}\big|\mathtt{g}_{nk}\big|^2+N_{0}\big)^{-1}\hspace{-.49cm}}.
			\end{align}
	The gain and the variance of the beamforming uncertainty related terms, i.e., the numerator term and the first term in the denominator of~\eqref{eq:R_dl_ergodic}, can be obtained via steps similar to those in the UL case. The second term in the denominator of~\eqref{eq:R_dl_ergodic}, which is the inter-UE  interference due to data streams of the UEs that share pilots with the $n$th UE, i.e., $q\in\mathcal{I}_{n}\backslash n$, can be expressed as
	\begin{align}\label{eq:MUI_interm_4}
	\textstyle{\mathbb{E}\big[\big|\sum\nolimits_{j\in\mathcal{A}_{d}}\kappa_{jq}\sqrt{\mathcal{E}_{d,j}}\mathbf{f}_{jn}^{T}}&\textstyle{\hat{\mathbf{f}}_{jq}^{*}\big|^{2}\big]=\mathbb{E}\big[\big|\sum\nolimits_{j\in\mathcal{A}_{d}}\kappa_{jq}\sqrt{\mathcal{E}_{d,j}}\tau_p\sqrt{\mathcal{E}_{p,n}\mathcal{E}_{p,q}}{c}_{jq}\beta_{jq}\|\mathbf{f}_{jn}\|^2\big|^{2}\big]}\notag\\&\hspace*{-4cm}\textstyle{+\mathbb{E}\big[\big|\sum\nolimits_{j\in\mathcal{A}_{d}}\kappa_{jq}\sqrt{\mathcal{E}_{d,j}}{c}_{jq}\sqrt{\tau_p\mathcal{E}_{p,q}}\beta_{jq}\mathbf{f}_{jn}^{T}\big(\sum\nolimits_{\substack{q'\in\mathcal{{I}}_{p}\backslash n}}\sqrt{\tau_{p}\mathcal{E}_{p,q'}}\mathbf{f}_{jq'}+\dot{\mathbf{w}}_{p,j}\big)\big|^2\big]}
	\notag\\&\hspace*{-4cm}\textstyle{=N(N+1)\tau_p^{2}\mathcal{E}_{p,n}\mathcal{E}_{p,q}\sum\nolimits_{j\in\mathcal{A}_{d}}\kappa_{jq}^2\mathcal{E}_{d,j}{c}_{jq}^2\beta_{jq}^{2}\beta_{jn}^2+N^2\tau_{p}^2\mathcal{E}_{p,n}\mathcal{E}_{p,q}\big(\sum\nolimits_{j\in\mathcal{A}_{d}}\sqrt{\mathcal{E}_{d,j}}\kappa_{jq}c_{jq}\beta_{jq}\beta_{jn}\big)}\notag\\&\hspace*{-4cm}\textstyle{\times\big(\!\!\!\!\!\!\sum\limits_{\substack{j'\in\mathcal{A}_{d},  j'\neq j}}\!\!\!\!\!\!\sqrt{\mathcal{E}_{d,j'}}\kappa_{j'q}c_{j'q}\beta_{j'q}\beta_{j'n}\big)+\!\! N\tau_{p}\mathcal{E}_{p,q}\sum\limits_{j\in\mathcal{A}_{d}}\kappa_{jq}^2\mathcal{E}_{d,j}{c}_{jq}^2\beta_{jq}^2\big(\!\!\sum_{q'\in\mathcal{{I}}_{p}\backslash n}{\tau_{p}\mathcal{E}_{p,q'}}\beta_{jq'}+N_{0}\big)\beta_{jn}}.
	\end{align}
	Further algebraic manipulations yield
	\begin{align}\label{eq:dl_sinr_coh}
	&\textstyle{\sum\nolimits_{\substack{q\in \mathcal{I}_{p}\backslash n}}\mathbb{E}\big[\big|\sum\nolimits_{j\in\mathcal{A}_{d}}\kappa_{jq}\sqrt{\mathcal{E}_{d,j}}\mathbf{f}_{jn}^{T}\hat{\mathbf{f}}_{jq}^{*}\big|^{2}\big]+\sum\nolimits_{{q\in \mathcal{U}_{d}\backslash \mathcal{I}_{p}}}\mathbb{E}\big[\big|\sum\nolimits_{j\in\mathcal{A}_{d}}\kappa_{jq}\sqrt{\mathcal{E}_{d,q}}\mathbf{f}_{jn}^{T}\hat{\mathbf{f}}_{jq}^{*}\big|^{2}\big]}\notag\\&\textstyle{=N^{2}\sum\nolimits_{\substack{q\in \mathcal{I}_{p}\backslash n}}\big(\sum\nolimits_{j\in\mathcal{A}_{d}}\sqrt{\mathcal{E}_{d,j}}\kappa_{jq}\alpha_{jq}^2\textstyle{\sqrt{\frac{\mathcal{E}_{p,n}}{\mathcal{E}_{p,q}}}}\frac{\beta_{jn}}{\beta_{jq}}\big)^{2}+N\sum\nolimits_{q\in \mathcal{U}_{d}\backslash n}\sum\nolimits_{j\in\mathcal{A}_{d}}\mathcal{E}_{d,j}\kappa_{jq}^2\beta_{jn}\alpha_{jq}^2.}
	\end{align}
	The first term of~\eqref{eq:dl_sinr_coh} equates to~\eqref{eq:dl_coh}. The second term of~\eqref{eq:dl_sinr_coh} together with the $n$th UE's beamforming uncertainty corresponds to~\eqref{eq:dl_ncoh}. Finally,~\eqref{eq:dl_UEUE} follows because $\mathbb{E}\big|\mathtt{g}_{nk}\big|^2=\epsilon_{nk}$. 
\end{proof}
	\vspace*{-.4cm}
\section{Proof of Theorem~\ref{thm:submod}}\label{proof:submod_thm}
\begin{proof}
	We present an inductive proof. Let us assume we schedule the APs in $\mathcal{A}_{s}$ such that $f_{mk}(\mathcal{A}_{s})$ is maximized. Now consider the set $\mathcal{A}_{t}$, such that $\mathcal{A}_{s}\subseteq\mathcal{A}_{t}$. We need to prove that, if we schedule any AP $\{j\}\notin\mathcal{A}_{t}$ next, the incremental gain attained by adding $\{j\}$ to $\mathcal{A}_{t}$ is smaller than the incremental gain achieved by adding $\{j\}$ to $\mathcal{A}_{s}$. Now, by our hypothesis, the set $\mathcal{A}_{s}$ is determined first to maximize $f_{mk}(.)$. Therefore, since AP $\{j\}$ is not part of $\mathcal{A}_{s}$, the product SINR under $\mathcal{A}_{s}$ is greater than that attained via only activating the $\{j\}$th AP in either of the mode of transmissions, that is $\prod_{k=1}^{K}\frac{\sum_{m\in\mathcal{A}_{s}}{\tt G}_{mk}(\mathcal{A}_{s})}{\sum_{m\in\mathcal{A}_{s}}{\tt I}_{mk}(\mathcal{A}_{s})}\geq \prod_{k=1}^{K}\frac{{\tt G}_{jk}(\{j\})}{{\tt I}_{jk}(\{j\})}$.
	Now using the monotonic nondecreasing property in Definition~\ref{defn:sub_mod}, we can write
	\begin{align}
	&\textstyle{{\prod_{k=1}^{K}\frac{\sum_{m\in\mathcal{A}_{t}}{\tt G}_{mk}(\mathcal{A}_{t})}{\sum_{m\in\mathcal{A}_{t}}{\tt I}_{mk}(\mathcal{A}_{t})}\geq \prod_{k=1}^{K}\frac{\sum_{m\in\mathcal{A}_{s}}{\tt G}_{mk}(\mathcal{A}_{s})}{\sum_{m\in\mathcal{A}_{s}}{\tt I}_{mk}(\mathcal{A}_{s})}}}\notag\\&\textstyle{\Rightarrow\frac{\prod_{k=1}^{K}\sum_{m\in\mathcal{A}_{t}}{\tt G}_{mk}(\mathcal{A}_{t})-\prod_{k=1}^{K}\sum_{m\in\mathcal{A}_{t}}{\tt G}_{mk}(\mathcal{A}_{s})}{\prod_{k=1}^{K}\sum_{m\in\mathcal{A}_{t}}{\tt I}_{mk}(\mathcal{A}_{t})-\prod_{k=1}^{K}\sum_{m\in\mathcal{A}_{s}}{\tt I}_{mk}(\mathcal{A}_{s})}\geq \frac{\prod_{k=1}^{K}\sum_{m\in\mathcal{A}_{s}}{\tt G}_{mk}(\mathcal{A}_{s})}{\prod_{k=1}^{K}\sum_{m\in\mathcal{A}_{s}}{\tt I}_{mk}(\mathcal{A}_{s})}\geq\prod_{k=1}^{K}\frac{{\tt G}_{jk}(\{j\})}{{\tt I}_{jk}(\{j\})}}
	\notag\\&\textstyle{\Rightarrow -\prod_{k=1}^{K}{{\tt I}_{jk}(\{j\})}\prod_{k=1}^{K}\sum_{m\in\mathcal{A}_{s}}{\tt G}_{mk}(\mathcal{A}_{s})+\prod_{k=1}^{K}{\tt G}_{jk}(\{j\})\prod_{k=1}^{K}\sum_{m\in\mathcal{A}_{s}}{\tt I}_{mk}(\mathcal{A}_{s})}\notag\\&\textstyle{\geq-\prod_{k=1}^{K}{{\tt I}_{jk}(\{j\})}\prod_{k=1}^{K}\sum_{m\in\mathcal{A}_{t}}{\tt G}_{mk}(\mathcal{A}_{t})+\prod_{k=1}^{K}{\tt G}_{jk}(\{j\})\prod_{k=1}^{K}\sum_{m\in\mathcal{A}_{t}}{\tt I}_{mk}(\mathcal{A}_{t})}\label{eq:sub_mod_addsub}
	\end{align}
	Next, adding and subtracting  $ \prod_{k=1}^{K}\sum_{m\in\mathcal{A}_{s}}{\tt G}_{mk}(\mathcal{A}_{s})\prod_{k=1}^{K}\sum_{m\in\mathcal{A}_{s}}{\tt I}_{mk}(\mathcal{A}_{s})$  on the left hand side and 
	$\prod_{k=1}^{K}\sum_{m\in\mathcal{A}_{t}}{\tt G}_{mk}(\mathcal{A}_{t})\prod_{k=1}^{K}\sum_{m\in\mathcal{A}_{t}}{\tt I}_{mk}(\mathcal{A}_{t})$ on the right hand side of~\eqref{eq:sub_mod_addsub}, we get
	\vspace*{-0.2cm}
	\begin{align}
	&\hspace{-1cm}\textstyle{-\prod_{k=1}^{K}\big({{\tt I}_{jk}(\{j\})}+\sum_{m\in\mathcal{A}_{s}}{\tt I}_{mk}(\mathcal{A}_{s})\big)\prod_{k=1}^{K}\sum_{m\in\mathcal{A}_{s}}{\tt G}_{mk}(\mathcal{A}_{s})}\notag\\&\textstyle{+\prod_{k=1}^{K}\big({\tt G}_{jk}(\{j\})+\sum_{m\in\mathcal{A}_{s}}{\tt G}_{mk}(\mathcal{A}_{s})\big)\prod_{k=1}^{K}\sum_{m\in\mathcal{A}_{s}}{\tt I}_{mk}(\mathcal{A}_{s})}\notag\\&\hspace{-1cm}\textstyle{\geq-\prod_{k=1}^{K}\big({{\tt I}_{jk}(\{j\})}+\sum_{m\in\mathcal{A}_{t}}{\tt I}_{mk}(\mathcal{A}_{t})\big)\prod_{k=1}^{K}\sum_{m\in\mathcal{A}_{t}}{\tt G}_{mk}(\mathcal{A}_{t})}\notag\\&\textstyle{+\prod_{k=1}^{K}\big({\tt G}_{jk}(\{j\})+\sum_{m\in\mathcal{A}_{t}}{\tt G}_{mk}(\mathcal{A}_{t})\big)\prod_{k=1}^{K}\sum_{m\in\mathcal{A}_{t}}{\tt I}_{mk}(\mathcal{A}_{t})},
	\end{align}
	\vspace*{-.5cm}
	which equivalently becomes
	\begin{align*}
	&\textstyle{-\!\!\prod\limits_{k=1}^{K}\sum\limits_{m\in\mathcal{A}_{s}\cup\{j\}}{\tt I}_{mk}(\mathcal{A}_{s}\cup\{j\})\prod\limits_{k=1}^{K}\sum\limits_{m\in\mathcal{A}_{s}}{\tt G}_{mk}(\mathcal{A}_{s})\!\!+\!\!\prod\limits_{k=1}^{K}\sum_{m\in\mathcal{A}_{s}\cup\{j\}}{\tt G}_{mk}(\mathcal{A}_{s}\cup\{j\})\prod\limits_{k=1}^{K}\sum\limits_{m\in\mathcal{A}_{s}}{\tt I}_{mk}(\mathcal{A}_{s})}\notag\\&\textstyle{\geq\!\!-\prod\limits_{k=1}^{K}\sum\limits_{m\in\mathcal{A}_{t}\cup\{j\}}{\tt I}_{mk}(\mathcal{A}_{t}\cup\{j\})\prod\limits_{k=1}^{K}\sum\limits_{m\in\mathcal{A}_{t}}{\tt G}_{mk}(\mathcal{A}_{t})\!\!+\!\!\prod\limits_{k=1}^{K}\sum\limits_{m\in\mathcal{A}_{t}\cup\{j\}}{\tt G}_{mk}(\mathcal{A}_{t}\cup\{j\})\prod\limits_{k=1}^{K}\sum\limits_{m\in\mathcal{A}_{t}}{\tt I}_{mk}(\mathcal{A}_{t})}.
	\end{align*}
	Using the fact that
		\vspace*{-.4cm} $$\textstyle{\prod\limits_{k=1}^{K}\sum\limits_{m\in\mathcal{A}_{s}\cup\{j\}}{\tt I}_{mk}(\mathcal{A}_{s}\cup\{j\})\prod\limits_{k=1}^{K}\sum\limits_{m\in\mathcal{A}_{s}}{\tt I}_{mk}(\mathcal{A}_{s})\leq\prod\limits_{k=1}^{K}\sum\limits_{m\in\mathcal{A}_{t}\cup\{j\}}{\tt I}_{mk}(\mathcal{A}_{t}\cup\{j\})\prod\limits_{k=1}^{K}\sum\limits_{m\in\mathcal{A}_{t}}{\tt I}_{mk}(\mathcal{A}_{t})},$$
	we can finally write,
		\vspace*{-.4cm}
	\begin{align*}
	&{\frac{\prod_{k=1}^{K}\sum_{m\in\mathcal{A}_{s}\cup\{j\}}{\tt G}_{mk}(\mathcal{A}_{s}\cup\{j\})}{\prod_{k=1}^{K}\sum_{m\in\mathcal{A}_{s}\cup\{j\}}{\tt I}_{mk}(\mathcal{A}_{s}\cup\{j\})}-\frac{\prod_{k=1}^{K}\sum_{m\in\mathcal{A}_{s}}{\tt G}_{mk}(\mathcal{A}_{s})}{\prod_{k=1}^{K}\sum_{m\in\mathcal{A}_{s}}{\tt I}_{mk}(\mathcal{A}_{s})}} \nonumber \\ &\hspace{1cm} {\geq\frac{\prod_{k=1}^{K}\sum_{m\in\mathcal{A}_{t}\cup\{j\}}{\tt G}_{mk}(\mathcal{A}_{t}\cup\{j\})}{\prod_{k=1}^{K}\sum_{m\in\mathcal{A}_{t}\cup\{j\}}{\tt I}_{mk}(\mathcal{A}_{t}\cup\{j\})}-\frac{\prod_{k=1}^{K}\sum_{m\in\mathcal{A}_{t}}{\tt G}_{mk}(\mathcal{A}_{t})}{\prod_{k=1}^{K}\sum_{m\in\mathcal{A}_{t}}{\tt I}_{mk}(\mathcal{A}_{t})}},
	\end{align*}
	which reduces to Theorem~\ref{thm:submod}.
\end{proof}
\end{appendices}
	\ifCLASSOPTIONcaptionsoff
\newpage
\fi
%\nocite{*}
\bibliographystyle{IEEEtran}
\bibliography{bib_cell_free}
\end{document}